\documentclass[sigconf,screen]{acmart}
\AtBeginDocument{%
  }

\usepackage{amsmath}
\usepackage{amsthm}
\usepackage{amsfonts}
\usepackage{thm-restate}
\usepackage{enumitem}
\usepackage{subcaption}

\usepackage[linesnumbered, boxed, vlined]{algorithm2e}
\DontPrintSemicolon
\SetKwInOut{Input}{Input}
\SetKwInOut{Output}{Output}
\SetKwProg{function}{Function}{:}{}
\SetKwProg{parfor}{parallel for}{do}{}

\usepackage{algorithmic}

\usepackage{graphicx}
\usepackage{tabularx}
\usepackage{multirow}
\usepackage{xspace}
\usepackage[capitalise]{cleveref}

\theoremstyle{plain}
\newtheorem{theorem}{Theorem}[section]

\newtheorem{problem}{Problem}
\newtheorem{definition}{Definition}[section]

\newtheorem{question}{Question}

\makeatletter
\newcommand{\multiline}[1]{%
  \begin{tabularx}{\dimexpr\linewidth-\ALG@thistlm}[t]{@{}X@{}}
    #1
  \end{tabularx}
}
\makeatother

\newcommand{\myparagraph}[1]{\vspace{.03in}\noindent {\bf #1.}}

\newcommand{\sldmerge}{\textsc{SLD-Merge}\xspace}

\newcommand{\ins}{\textsc{Insert}}
\newcommand{\del}{\textsc{Delete}}

\newcommand{\filter}{\textsc{Filter}}
\newcommand{\dynsld}{\textsc{DynSLD}\xspace}

\newcommand{\spine}{\mathsf{Spine}\xspace}

\def \tO {\tilde{O}\xspace}

\newcommand{\defn}[1]{{\textit{\textbf{\boldmath #1}}}}

\newcounter{mnote}[section]

\newcommand{\revision}[1]{#1}
\begin{document}
\fancyhead{}

\title{Fully-Dynamic Parallel Algorithms for Single-Linkage Clustering}

\author{Quinten De Man}
\affiliation{
    \institution{University of Maryland}
    \city{College Park}
    \state{MD}
    \country{USA}
}
\email{deman@umd.edu}

\author{Laxman Dhulipala}
\affiliation{
    \institution{University of Maryland}
   \city{College Park}
   \state{MD}
   \country{USA}
}
\email{laxman@umd.edu}

\author{Kishen N Gowda}
\affiliation{
    \institution{University of Maryland}
   \city{College Park}
   \state{MD}
   \country{USA}
}
\email{kishen19@cs.umd.edu}

\begin{abstract}
Single-linkage clustering is a popular form of hierarchical agglomerative clustering (HAC) where the distance between two clusters is defined as the minimum distance between any pair of points across the two clusters. In single-linkage HAC, the output is typically the single-linkage dendrogram (SLD), which is the binary tree representing the hierarchy of clusters formed by iteratively contracting the two closest clusters. In the dynamic setting, prior work has only studied maintaining a minimum spanning forest over the data since single-linkage HAC reduces to computing the SLD on the minimum spanning forest of the data.

In this paper, we study the problem of maintaining the SLD in the fully-dynamic setting. We assume the input is a dynamic forest $F$ (representing the minimum spanning forest of the data) which receives a sequence of edge insertions and edge deletions. To our knowledge, no prior work has provided algorithms to update an SLD asymptotically faster than recomputing it from scratch. All of our update algorithms are asymptotically faster than the best known static SLD computation algorithm, which takes $O(n \log h)$ time where $h$ is the height of the dendrogram ($h \leq n-1$). Furthermore, our algorithms are much faster in many cases, such as when $h$ is low. Our first set of results are an insertion algorithm in $O(h)$ time and a deletion algorithm in $O(h \log (1+n/h))$ time. Next, we describe parallel and batch-parallel versions of these algorithms which are work-efficient or nearly work-efficient and have poly-logarithmic depth. Finally, we show how to perform insertions near-optimally in $O(c \log(1+n/c))$ time, where $c$ is the number of structural changes in the dendrogram caused by the update, and give a work-efficient parallel version of this algorithm that has polylogarithmic depth.

\end{abstract}




\renewcommand\footnotetextcopyrightpermission[1]{} 

\copyrightyear{2025}
\acmYear{2025}
\setcopyright{cc}
\setcctype{by}
\acmConference[SPAA '25]{37th ACM Symposium on Parallelism in Algorithms and Architectures}{July 28-August 1, 2025}{Portland, OR, USA}
\acmBooktitle{37th ACM Symposium on Parallelism in Algorithms and Architectures (SPAA '25), July 28-August 1, 2025, Portland, OR, USA}
\acmDOI{10.1145/3694906.3743315}
\acmISBN{979-8-4007-1258-6/2025/07}

\maketitle

\section{Introduction}\label{sec:intro}
Single-linkage clustering is a popular form of {\em hierarchical agglomerative clustering} (HAC), a fundamental clustering algorithm that has been widely used in unsupervised learning and data mining~\cite{irbook}. 
In HAC, clusters are greedily paired together based on a distance (or similarity) measure. 
For single-linkage clustering, the distance between clusters is determined by the {\em linkage function} $D(A,B) = \min_{a\in A,b\in B} d(a,b)$.
The output of HAC is typically a {\em dendrogram}, a rooted binary tree where the leaves are clusters containing each original point, and each internal node corresponds to merging a pair of clusters. 
For single-linkage clustering, this tree is called the {\em single-linkage dendrogram} (SLD).

Single-linkage dendrograms have been widely used to analyze real-world data in fields ranging from computational biology~\cite{gasperini2019genome, letunic2007interactive, yengo2022saturated},
image analysis~\cite{gotz2018parallel, havel2019efficient, ouzounis2012alpha}, and astronomy~\cite{baron2019machine, feigelson1998statistical}, among many
others~\cite{henry2005cluster, yim2015hierarchical, irbook}. 
Due to its real-world importance and its role as a sub-step in other fundamental clustering algorithms such as HDBSCAN~\cite{campello2015hierarchical, wang2020theoretically, wang2021fast}, computing the SLD of an input weighted
tree has been widely studied by parallel algorithms researchers in
recent years, with novel algorithms and implementations being
proposed for the shared-memory setting~\cite{havel2019efficient, wang2021fast}, GPUs~\cite{nolet2023cuslink, sao2024pandora},
and distributed memory settings~\cite{gotz2018parallel, havel2019efficient}.
We also observe that single linkage dendrograms are closely related to {\em Cartesian trees}~\cite{vuillemin1980unifying} 
, and we formalize the connection in \revision{Section~\ref{sec:applications} of} this paper.

Traditionally, HAC algorithms have been {\em point-based}~\cite{yu2021parchain}, meaning each pair of points has some distance (or similarity).
Taken together, the distances compose a distance matrix.
Guided by the observation that this matrix may have many negligible entries in some scenarios, recent studies~\cite{dhulipala2021hierarchical, tseng2022parallel, parhac, yu2025dynhac, bateni2024s, dhulipala2023terahac} have focused on {\em graph-based} HAC algorithms. 
In this setting, the input is a graph which only defines similarities for a subset of the pairs.

Due to the rapidly changing nature of modern datasets, it is of significant interest to obtain practical dynamic HAC algorithms that can dynamically maintain a clustering as the underlying data changes.
Similar to the static setting, graph-based HAC algorithms are more suitable for fast dynamic algorithms as an insertion or deletion of a point results in a linear number of changes to the similarity matrix in point-based HAC.
Indeed, prior studies of HAC in the dynamic setting have focused only on graph-based HAC~\cite{tseng2022parallel, yu2025dynhac}.
For dynamic graph-based HAC, the input is a dynamic graph subject to edge insertions, deletions, or similarity updates.

Recently, there has been significant effort to develop parallel batch-dynamic algorithms which process \emph{batches} of updates or queries for fundamental problems on graphs, ordered sets and maps, and other data structures~\cite{acar2019batchconnect, anderson2024deterministic, anderson2023deterministic, tseng2022parallel, wang2020closest, yesantharao2021parallel, acar2020changeprop, dhulipala2019parallel, de2024towards, liu2022parallel, tseng2019batch, dhulipala2019low, dhulipala2021parallel, dhulipala2022pac, men2025parallel}.
Compared to sequential dynamic algorithms, such algorithms have the potential to be significantly more applicable in practice as they are not forced to apply the updates or queries one at a time.
Motivated by the problem of parallelizing dynamic graph-based single-linkage clustering, Tseng et al.~\cite{tseng2022parallel} developed a parallel batch-dynamic minimum spanning forest (MSF) algorithm.
It is well known that single-linkage clustering can be reduced to computing a minimum spanning tree on the points~\cite{Gower1969MST}, or in the case of graph-based HAC, the MSF of the graph.
However, the MSF does not directly provide the SLD, and computing the SLD requires non-trivial processing of the MSF. 
Although dynamically maintaining the MSF of the graph is a crucial step in an end-to-end pipeline for dynamic SLD, Tseng et al. do not maintain the SLD itself. 
Instead, their algorithm only supports pairwise {\em threshold queries} of the form ``are two vertices in the same cluster if we agglomeratively cluster all edges with similarity above a certain threshold?''
Although threshold queries may be sufficient for some applications, solely maintaining an MSF lacks the full power of maintaining an {\em explicit single-linkage dendrogram}.
For example, an explicit SLD can efficiently report all vertices in the cluster defined by a particular edge, or return a flat clustering of the data in low depth. 
In Section~\ref{sec:applications}, we describe several operations and applications that require an explicit SLD.

Prior algorithms for static single-linkage clustering all produce dendrograms, but there is no prior work that studies explicitly maintaining the dendrogram in the dynamic or batch-dynamic setting: there is a gap in the literature between dynamic MSF computation and dynamic SLD maintenance.
In this paper, we are interested in understanding and bridging this gap. 
Specifically, we study the problem of maintaining a data structure $D$ that explicitly represents the SLD of a dynamic graph $G$. 
The input to the problem is the MSF $F$ of the graph $G$, where the algorithm receives edge insertions or deletions in $F$ and must appropriately update $D$.

Explicitly maintaining the SLD is challenging, as there exist graphs where a single change to the MSF can completely change the structure of the dendrogram, implying that the worst-case update cost must be $\Omega(n)$ in some cases. 
In such cases, the cost of updating the dendrogram may be comparable to the cost for statically recomputing the SLD. 
For this reason, prior work~\cite{tseng2022parallel} did not attempt to develop algorithms for explicit SLD maintenance.
However, a more recent study~\cite{Dhulipala2024Dendrogram} provides valuable insight into the structural properties of dendrograms, revealing that the number of changes can actually be bounded by $h$, the height of the dendrogram. 
This fact indicates that in cases where $h = o(n)$, the number of structural changes induced by an update may be low, and there is hope for dynamically updating the dendrogram at a lower cost than static recomputation. This motivates our first research question:

\begin{question}\label{question:height}
    Given an MSF $F$, a data structure $D$ explicitly representing the SLD of $F$, and an edge update which changes $F$ to $F'$, can we update $D$ to represent the SLD of $F'$ in time proportional to $h$, the height of the dendrogram?
\end{question}

Even if $h$ is in $\Omega(n)$, it is still possible for the number of changes to the dendrogram structure caused by a single update to be smaller. Therefore, we are also interested in developing {\em output-sensitive} update algorithms, which do work proportional to the actual number of structural changes caused by an update. Specifically, we study the following second research question:

\begin{question}\label{question:change}
    Given an MSF $F$, a data structure $D$ explicitly representing the SLD of $F$, and an edge update which changes $F$ to $F'$, can we update $D$ to represent the SLD of $F'$ in time proportional to $c$, the number of structural changes in the dendrogram?
\end{question}

In the case where both $h$ and $c$ are large, static recomputation of the SLD would still be the more effective solution, since it is known how to parallelize static SLD computation work-efficiently and in low depth~\cite{Dhulipala2024Dendrogram}.
Therefore, the final focus of our study is developing parallel algorithms for updating the explicit SLD.
We are also interested in parallel batch-dynamic algorithms for explicitly maintaining the SLD, as combining such an algorithm with the parallel batch-dynamic MSF algorithm of Tseng et al.~\cite{tseng2022parallel} would yield an end-to-end parallel batch-dynamic algorithm for graph-based single-linkage clustering. This yields our third research question:

\begin{question}\label{question:parallel}
    Given a dynamic MSF $F$ subject to single updates or batches of updates, can we update an explicit representation $D$ of the dendrogram work-efficiently and in polylogarithmic depth?
\end{question}

\subsection{Our Contributions}
This paper describes several theoretically efficient algorithms for explicit SLD maintenance given a dynamic MSF as input.
\revision{Our results are summarized as follows:}
\revision{
\begin{itemize}[itemsep=0pt,parsep=0pt,leftmargin=15pt]
    \item An insertion algorithm in $O(h)$ time and a deletion algorithm in $O(h \log (1+n/h))$ time, where $h$ is the dendrogram height.
    \item An output-sensitive insertion algorithm in $O(c \log(1+n/c))$ time, where $c$ is the number of structural changes in the dendrogram caused by the insertion.
    \item Parallel versions of all the aforementioned update algorithms which are work-efficient or nearly work-efficient and run in poly-logarithmic depth. Specifically, our parallel insertion and deletion algorithms both take $O(h \log (1+n/h))$ work and $O(\log n \log h)$ depth. Our parallel output-sensitive insertion algorithm takes $O(c \log(1+n/c))$ work and $O(\log n \log h)$ depth.
    \item Parallel batch-dynamic algorithms for batches of insertions and batches of deletions that are work-efficient and have poly-logarithmic depth. Specifically, our batch insertion algorithm takes $O(kh \log(1+n/(kh)))$ work and $O(\log n \log k \log(kh))$ depth. Our batch deletion algorithm takes $O(kh \log(1+n/(kh)))$ work and $O(\log n \log(kh))$ depth.
\end{itemize}
}
\revision{All of our algorithms are deterministic and our parallel algorithms are analyzed in the binary fork-join model~\cite{blelloch2019optimal}. We now provide a more detailed discussion of each of these results.}

\myparagraph{Height Bounded SLD Updates}
\revision{Our first set of results includes a sequential insertion algorithm in $O(h)$ time and a sequential deletion algorithm in $O(h \log (1+n/h))$ time (Theorem~\ref{thm:heightalgs}).}
Insertions can be processed in $O(h)$ time by employing the {\em spine-merge} technique~\cite{Dhulipala2024Dendrogram} which was originally developed for parallel merge-based \emph{static} SLD computation;
\revision{we summarize this technique in Section~\ref{sec:dynsldsequential}.}
For edge deletions, we show that the structural changes to the dendrogram can be thought of as {\em unmerging the spine} corresponding to the deleted edge\revision{, where the {\em spine} of a node in a dendrogram corresponds to its node-to-root path}. 
Our deletion algorithm accomplishes this by maintaining a dynamic tree data structure (for tree connectivity) on the input forest, and separating each node on the spine based on which side of the cut induced by the edge deletion it falls (using a batch connectivity query).

Theorem~\ref{thm:heightalgs} affirmatively answers Question~\ref{question:height}, giving a sequential algorithm that can process updates with cost proportional to the height of the dendrogram. 
Although $h$ can be as large as $n-1$ in the worst-case, this algorithm is already of interest as it will always (both for insertions and deletions) be asymptotically better than static recomputation; an optimal static algorithm requires $\Theta(n\log h)$ time~\cite{Dhulipala2024Dendrogram}, which, since $h \leq n-1$, is asymptotically larger than $O(h)$ for insertions, and asymptotically larger than $O(h \log(1+n/h))$ for deletions (maximized when $h=n-1$ and taking a value of $O(n)$).
Additionally, if $h$ is smaller, the asymptotic cost becomes much smaller than static recomputation (e.g., for $h = \Theta(\log n)$, insertions cost $O(\log n)$ and deletions cost $O(\log^2 n)$).

\begin{restatable}[\revision{Sequential Updates}]{theorem}{heightalgs}\label{thm:heightalgs}
    There exists a dynamic SLD algorithm that processes edge insertions in $O(h)$ time and edge deletions in $O(h \log(1+n/h))$ time, where $h$ is the height of the dendrogram.
\end{restatable}

\myparagraph{Near Optimal Output-Sensitive Insertions}
\revision{Our next result is an insertion-only algorithm in $O(c \log(1+n/c))$ time, where $c$ is the actual number of structural changes in the dendrogram caused by the update (Theorem~\ref{thm:changealgs}).}
\revision{We accomplish this by} improving the efficiency of spine merging using a new kind of path query on dynamic trees called a {\em path weight search} (PWS) query (we show how to implement PWS queries for rake-compress trees~\cite{acar2004dynamizing, acar2005experimental}). A PWS query finds the largest weight node in a path with weight less than some threshold.
Using PWS queries, the algorithm can directly determine the new parent of a given node in the spine in $O(\log n)$ time, rather than traversing through all the nodes in both spines. 
By carefully alternating between the two spines, we show that the merge can be completed with exactly $c$ PWS queries and $c$ pointer changes. \revision{We then show how to combine these queries to avoid redundant work and achieve the time bound from our result.}

Theorem~\ref{thm:changealgs} partially answers Question~\ref{question:change}: our algorithm can process insertions in time proportional to the number of structural changes to the dendrogram. However, achieving a similar bound for deletions remains an open question.

Our output-sensitive incremental algorithm is near-optimal in the sense that the update cost is within a logarithmic factor of the trivial $\Omega(c)$ lower bound for any update algorithm that explicitly maintains the SLD structure.
In the case where $c=O(1)$, there is a simple amortized lower bound of $\Omega(\log n)$ per operation based on the $\Omega(n \log n)$ lower bound for static SLD computation and comparison-based sorting, and thus, our algorithm is optimal in this case.
Furthermore, this output-sensitive algorithm for insertions can be used in conjunction with the deletion algorithm from Theorem~\ref{thm:heightalgs}, resulting in a fully-dynamic SLD algorithm that is near-optimal for insertions, and always asymptotically better than static recomputation.

\begin{restatable}[\revision{Sequential Output-Sensitive Updates}]{theorem}{changealgs}\label{thm:changealgs}
    There exists a dynamic SLD algorithm that processes edge insertions in $O(c \log(1+n/c))$ time, where $c$ is the number of structural changes to the dendrogram caused by the update.
\end{restatable}

\myparagraph{Parallel Updates}
\revision{Our next set of results are parallel versions of the aforementioned algorithms which are work-efficient or nearly work-efficient and have poly-logarithmic depth.}
A critical step in our parallel update algorithms is to access all of the nodes in a spine in low depth. We do this by maintaining a rake-compress tree~\cite{acar2004dynamizing, acar2005experimental} of the SLD, and traversing down to the nodes in the spine using the {\em path decomposition} in the rake-compress tree. For insertions, we then develop a parallel spine-merge technique that works similarly to classic parallel algorithms for merging two sorted arrays. 
For deletions, we perform a parallel batch of connectivity queries for each node in the spines, perform a parallel filter to separate the node into two sides, and then update all the parent-pointers in parallel.
To develop a parallel output-sensitive incremental algorithm, we define and use {\em path median} and {\em path weight search} queries, and design a novel divide-and-conquer technique to merge two spines.

Theorem~\ref{thm:parallelalgs}~and~\ref{thm:parallelchangealg} affirmatively answers the first part of Question~\ref{question:parallel}: we can in fact perform single updates work-efficiently and in polylogarithmic depth with our algorithms.
Not only is this interesting theoretically, but it enables our dynamic algorithms to be competitive with static recomputation, which has already been made highly parallel and practical~\cite{Dhulipala2024Dendrogram}.

\begin{restatable}[\revision{Parallel Updates}]{theorem}{parallelalgs}\label{thm:parallelalgs}
    There exists a \revision{parallel} dynamic SLD algorithm that processes edge insertions and edge deletions in $O(h \log(1+n/h))$ work and $O(\log n \log h)$ depth in the binary fork-join model, where $h$ is the height of the dendrogram.
\end{restatable}

\begin{restatable}[\revision{Parallel Output-Sensitive Updates}]{theorem}{parallelchangealg}\label{thm:parallelchangealg}
    There exists a \revision{parallel} dynamic SLD algorithm that processes edge insertions in $O(c \log (1+n/c))$ work and $O(\log n \log h)$ depth, where $c$ is the number of structural changes to the dendrogram caused by the insertion.
\end{restatable}

\myparagraph{Batch-Parallel Updates}
\revision{Our final set of results are parallel batch-dynamic algorithms for insertions and deletions that are work-efficient and have poly-logarithmic depth (Theorem~\ref{thm:batchalgs}).}
To extend the parallel deletion algorithm to the batch setting, we discovered that a correct algorithm can be obtained by essentially running the single parallel deletion algorithm concurrently for each deleted edge. To achieve a better work bound, we batch all of the connectivity queries and RC tree updates.
The algorithm for batch insertions is much more complex. First, we show how to solve the problem efficiently in the case where all edges connect to one shared center component in a star pattern. In this case, all of the spines are ``merged into'' the center component. Using this {\em star merge} technique, we perform multiple rounds of tree contraction~\cite{miller1985parallel} until all components have been merged appropriately.

Theorem~\ref{thm:batchalgs} affirmatively answers the second part of Question~\ref{question:parallel}: we can also process batches of updates in polylogarithmic depth with our algorithms.
Our parallel batch-dynamic algorithms are interesting since if $h$ is very low, the single update parallel algorithms may not have enough work to take full advantage of multiple cores.
Additionally, combining this result with the parallel batch-dynamic MSF algorithm of Tseng et al. yields an end-to-end parallel batch-dynamic algorithm for SLD maintenance on a dynamic graph.

\begin{restatable}[\revision{Batch-Parallel Updates}]{theorem}{batchalgs}\label{thm:batchalgs}
    There exists a \revision{batch-parallel} dynamic SLD algorithm that processes batches of $k$ edge insertions in $O(k h \log (1+n/(kh)))$ work and $O(\log n \log k \log (kh))$ depth and batches of $k$ edge deletions in $O(kh \log (1+n/(kh)))$ work and $O(\log n \log(kh))$ depth, where $h$ is the height of the dendrogram.
\end{restatable}

\section{Preliminaries}\label{sec:prelims}

\subsection{Single-Linkage Clustering}\label{sec:prelimsld}

\begin{figure*}[t]
    \centering
    \includegraphics[width=0.9\linewidth]{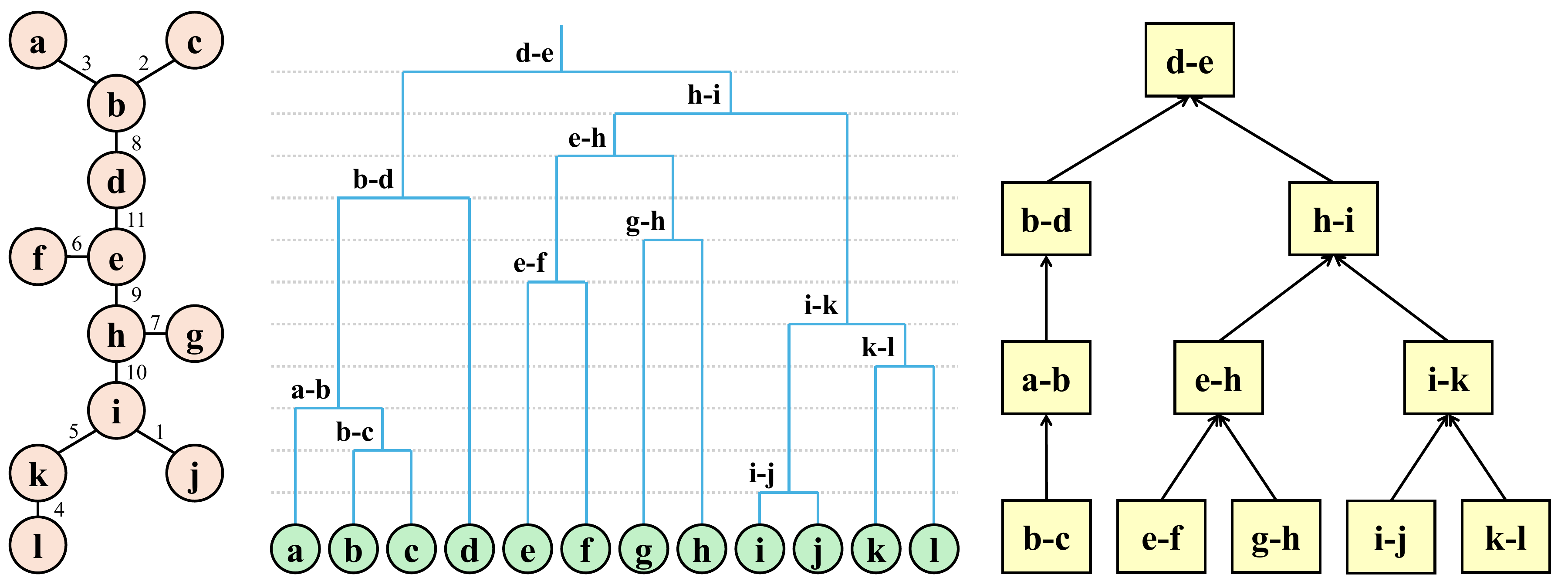}
    \caption{\small An example of a single-linkage dendrogram for a tree. The left shows the input tree with weighted edges; the edges are merged in order of increasing rank by the sequential algorithm. The middle shows a typical visualization of a dendrogram where the ``height'' of each edge corresponds to its rank or weight. The right shows the SLD data structure which only stores nodes for edges and parent-pointers.}
    \label{fig:sld_example}
\end{figure*}

Consider a weighted undirected graph $G = (V,E)$, where the weight of edge $e\in E$ is denoted by $w(e)$. In \defn{single-linkage clustering}, the distance $D(A,B)$ between two clusters $A$ and $B$ is the minimum weight of an edge in the cut induced on $(A,B)$:
\begin{equation*}
    D(A,B) = \min_{e \in \textsf{Cut}(A,B)} w(e).
\end{equation*}

In this paper, we assume the input graph is an edge-weighted tree as it is well known that single-linkage clustering on weighted graphs can be reduced to single-linkage clustering on the minimum spanning forest of the graph~\cite{Gower1969MST}.
Given an input edge-weighted forest $F = (V,E)$, let the \defn{rank} of an edge $e$, $r_e \in [n-1]$, be the position of this edge in the edge sequence sorted by weight (with ties broken consistently). We note that our algorithms do not require us to compute the ranks, however, using rank simplifies the presentation of our algorithms. We use the terms \emph{rank} and \emph{weight} interchangeably when referring to an edge.

The \defn{single-linkage dendrogram} (SLD) of a forest $F = (V,E)$ is a rooted binary forest $D$, where each leaf corresponds to a vertex in $V$ and each internal node corresponds to an edge in $E$. We denote the internal node corresponding to edge $e$ as $node(e)$ or simply node $e$.
\revision{The structure of the rooted binary forest corresponds to the merges of clusters that would occur by iteratively pairing together the closest distance clusters (using single-linkage clustering).}
We assume the output SLD will be stored as a \revision{binary tree}, where each node $e$ \revision{stores a pointer} to its parent node $p(e)$. For convenience, we drop the leaf nodes and only consider the tree on the $n-1$ internal edge nodes. For an edge $e\in E$, the \defn{spine} of $e$, denoted $\spine_D(e)$, is defined as the linked list starting from node $e$ until the root in $D$. Observe that the ranks of edges corresponding to the nodes along the spine of $e$ are in increasing order from node $e$ to the root. We use $\spine(e)$ when the tree and SLD are clear from context. We use the terms SLD and dendrogram interchangeably.
Figure~\ref{fig:sld_example} shows an example of an SLD for a given input tree.

\subsection{The Dynamic SLD Problem}
Here, we formally introduce and define the \defn{dynamic single-linkage dendrogram} problem, which we study in this paper. 
To the best of our knowledge, this is the first work to define and study this problem.
The goal is to {\em explicitly} maintain the SLD of an input {\em tree} subject to edge insertions and deletions.

\begin{problem}[The Fully-Dynamic Single-Linkage Dendrogram Problem]
    Given a dynamic weighted forest $F$ as input, explicitly maintain the single-linkage dendrogram for $F$. That is, maintain a data structure $D$ which explicitly represents the dendrogram as a binary tree. The input is a sequence of edge insertions or edge deletions in $F$.
\end{problem}

Similar to works on static graph-based SLD computation~\cite{Dhulipala2024Dendrogram}, we formulate the input as a dynamic forest rather than a dynamic graph because of the well known reduction from single-linkage clustering to single-linkage clustering on the minimum spanning forest (MSF) of the graph~\cite{Gower1969MST}.
Thus, a natural related problem is what we call the \defn{dynamic single-linkage clustering} problem. In this problem, the input is a dynamic {\em graph} subject to edge insertions and deletions.

\begin{problem}[The Fully-Dynamic Single-Linkage Clustering Problem]
    Given a dynamic weighted graph $G$ as input, explicitly maintain the single-linkage dendrogram for $G$. That is, maintain a data structure $D$ which explicitly represents the dendrogram as a binary tree. The input is a sequence of edge insertions or edge deletions in $G$.
\end{problem}

We do not directly study this second problem in this paper, but it is easy to see that combining our algorithms with existing dynamic MSF algorithms~\cite{holm2001poly, wulff2017fully} or batch-dynamic MSF algorithms~\cite{tseng2022parallel} results in solutions to the dynamic single-linkage clustering problem.

\subsection{Parallel Model and Primitives}\label{sec:parprelims}
We analyze all of our parallel algorithms in the \defn{binary fork-join model}~\cite{blelloch2019optimal}.
The binary fork-join model extends the standard RAM model with a {\em fork} instruction, which forks a child thread, and a {\em join} instruction to join threads after forking them. All threads share a common memory. A computation starts with a single initial thread, and ends when an {\em end all} operation is called. The computation can be viewed as a directed acyclic graph (DAG) where each node represents an instruction, nodes point to the node for the next instructions, nodes for fork instructions point to two nodes, and nodes for join instructions are pointed to by the last instruction of the child thread. The \defn{work} of a computation is defined as the total number of nodes in the DAG. The \defn{depth} of a computation is defined as the length of the longest path in the DAG.

We now define some parallel primitives used in this paper.
A \defn{parallel filter}~\cite{jaja1992parallel} takes a sequence of $n$ elements and some predicate on these elements and returns a new sequence of the elements for which the predicate returns true. Existing methods ensure that the ordering of elements is preserved in the filtered sequence, which our algorithms require. Parallel filter can be done in $O(n)$ work and $O(\log n)$ depth.
A \defn{parallel merge}~\cite{Cole1988, jaja1992parallel} takes two sequences in sorted order with $n$ total elements, and merges them into a single sequence in sorted order. Parallel merge can be done in $O(n)$ work and $O(\log n)$ depth.

\subsection{Dynamic Trees}\label{sec:dyntrees}

Dynamic trees~\cite{sleator1983data} are a fundamental data structure which are useful for efficiently processing certain queries on trees subject to edge insertions and deletions.
There exist many dynamic trees data structures with varying properties~\cite{sleator1983data, frederickson1985data, henzinger1995randomized}. In this paper, we focus only on rake-compress trees (RC trees)~\cite{acar2004dynamizing, acar2005experimental} because they can support theoretically efficient batch-parallel updates and queries~\cite{acar2020changeprop, anderson2023thesis, anderson2024deterministic}, and they provide a useful {\em path decomposition} property which our algorithms utilize. 
We use the deterministic version of parallel RC trees~\cite{anderson2024deterministic}, \revision{although there exist randomized versions with slightly better depth bounds for batch links and batch cuts~\cite{acar2020changeprop, anderson2023thesis}}. The remainder of this section summarizes RC trees.

\defn{Rake-compress trees} are a dynamic trees data structure based on the dynamization of static parallel tree contraction. The data structure maintains a hierarchical structure representing multiple rounds of contraction on the input tree. In each round, a maximal independent set of degree $1$ and $2$ vertices are contracted. When vertices and edges contract, they form {\em clusters} which are represented by vertices or edges in the tree at the next round.
Degree $1$ vertices \defn{rake} onto their neighboring vertex and form a cluster with their one incident edge. Degree $2$ vertices \defn{compress} with their two incident edges forming a new cluster that is the union of the clusters represented by the vertex and its two incident edges. This new cluster is represented by a single new edge in the next tree. After $O(\log n)$ rounds, the tree is fully contracted.

The RC tree represents this contraction process as a tree with height $O(\log n)$, where the leaves are the original edges and vertices in the input tree, and the children of each internal node are the clusters, edges, or vertices which were combined to form the cluster represented by the given internal node.
Importantly, there are two types of internal clusters in the RC tree.
\defn{Unary clusters} are formed by the ``rake'' of a degree $1$ vertex and represent some rooted subtree of the input tree.
\defn{Binary clusters} are formed from the ``compress'' of a degree $2$ vertex and represent the path between (but not including) two vertices as well as everything that branches off that path. This path is called the \defn{cluster path} of a binary cluster.
One important property is that any binary cluster is the combination of two binary clusters (or original edges) that were incident to a common vertex, and the cluster path of the parent cluster is the union of the cluster paths of these two children.
A second property is that for any two vertices $u$ and $v$, the RC tree can return a set of $O(\log n)$ binary clusters such that the union of their disjoint cluster paths forms exactly the path between $u$ and $v$. This set of clusters is called a \defn{path decomposition} and takes $O(\log n)$ time to compute.

In this paper, we use these properties to develop new operations on RC trees which allow us to efficiently operate on spines in the dendrogram dynamically and in parallel.
\revision{We omit most of the implementation details about path decompositions in RC trees, and refer to the work of Anderson for detailed explanations~\cite{anderson2023thesis}.}

We now describe the key operations on RC trees used in this paper. Table~\ref{tab:rc_tree_costs} summarizes the costs of each of these operations. For each operation, there is a corresponding batch parallel-operation. \revision{The depth bounds shown in the table are in the binary fork-join model. Note that the original algorithms for batch-parallel links and cuts were analyzed in the CRCW-PRAM model, but the algorithms can be easily translated to binary fork-join algorithms with the slightly larger depth bound shown in the table.}
\begin{itemize}[itemsep=0pt,parsep=0pt,leftmargin=15pt]
    \item \defn{$\mathsf{Link}(u,v)$}: Given two disconnected vertices $u$ and $v$, insert the edge $(u,v)$ joining the two tree components containing them.
    \item \defn{$\mathsf{Cut}(u,v)$}: Given two neighboring vertices $u$ and $v$, delete the edge $(u,v)$ splitting the tree containing them into two trees.
    \item \defn{$\mathsf{ConnectivityQuery}(u,v)$}: Returns whether vertices $u$ and $v$ are connected. Alternatively returns some representative value for the components containing $u$ and/or $v$.
    \item \defn{$\mathsf{PathQuery}(u,v)$}: Returns some aggregate information about the edge or vertex weights on the path from $u$ to $v$. Alternatively, returns the path decomposition (set of binary clusters) representing the path between $u$ and $v$.
\end{itemize}

\begin{table}[t]
    \centering
    \begin{tabular}{|c|c|c c|}
    \hline
        \multirow{2}{*}{Operation} & \multirow{2}{*}{Sequential} & \multicolumn{2}{c|}{Batch-Parallel} \\
        & & Work & Depth \\
        \hline\hline
        \multirow{2}{*}{Link} & \multirow{2}{*}{$O(\log n)$} & \multirow{2}{*}{$O(k \log(1+n/k))$} & \multirow{2}{*}{$O(\log n \log k)$} \\
         & & & \\
        \multirow{2}{*}{Cut} & \multirow{2}{*}{$O(\log n)$} & \multirow{2}{*}{$O(k \log(1+n/k))$} & \multirow{2}{*}{$O(\log n \log k)$} \\
         & & & \\
        Conn. & \multirow{2}{*}{$O(\log n)$} & \multirow{2}{*}{$O(k \log(1+n/k))$} & \multirow{2}{*}{$O(\log n)$} \\
        Query & & & \\
        Path  & \multirow{2}{*}{$O(\log n)$} & \multirow{2}{*}{$O(k \log n)$} & \multirow{2}{*}{$O(\log n)$} \\
        Query & & & \\
        \hline
    \end{tabular}
    \caption{\small The costs of various operations for RC trees on a tree with $n$ vertices, and batches of size $k$.}
    \label{tab:rc_tree_costs}
\end{table}

\section{\texorpdfstring{\dynsld}{DynSLD}: Maintaining Dendrograms Explicitly in the Fully-Dynamic Setting}\label{sec:dynsld}
In this section, we describe our algorithms which solve the dynamic single-linkage dendrogram problem. We broadly refer to our algorithms as \dynsld.
In Section~\ref{sec:dynsldsequential}, we begin by describing simple sequential update algorithms. In Section~\ref{sec:dynsldparallel}, we show how to derive parallel versions of our algorithms using RC trees.
In Section~\ref{sec:dynsldbatch}, we describe parallel batch-dynamic update algorithms.

\subsection{Sequential Update Algorithms}\label{sec:dynsldsequential}

\myparagraph{Spine-Merge Primitive} 
The edge insertion algorithm is based on the \sldmerge subroutine from Dhulipala et al.~\cite{Dhulipala2024Dendrogram}. We show the pseudo-code for \sldmerge in Algorithm~\ref{alg:sld-merge}.
Given two subtrees of the input $T_1$ and $T_2$ that share exactly one common vertex $v$ but no common edges, along with their respective SLDs $D_1$ and $D_2$, the SLD of $T_1\cup T_2$ can be obtained by merging the spines of the \emph{minimum-rank edges} incident to $v$ in both subtrees (see definitions of spine and rank in Section~\ref{sec:prelimsld}). We refer to these two spines as the \defn{ characteristic spines} corresponding to this merge.
{\em Merging two spines} is defined as combining the two sequences of nodes (which are ordered by increasing rank) into one such that the new sequence is also ordered by increasing rank.
Intuitively, all other nodes outside of these two spines in the SLD remain unaffected, as the clustering process continues to merge edges in the same (relative) order as before until it encounters a cluster that could potentially interact with clusters from the other subtree. However, such a cluster must contain the common vertex $v$, meaning that the spines of the minimum-rank edges incident to $v$ precisely correspond to the set of clusters that could potentially interact with those from the other subtree during clustering.

\begin{algorithm}
\caption{Merging two single-linkage dendrograms.}\label{alg:sld-merge}
\DontPrintSemicolon
\function{$\sldmerge(T_1, T_2, v)$}{
    Let $e_1^\star$ denote the edge with minimum rank incident to vertex $v$ in $T_1$.\\
    Let $e_2^\star$ denote the edge with minimum rank incident to vertex $v$ in $T_2$.\\
    Merge $D_1$ and $D_2$ along the spines of $e_1^\star$ and $e_2^\star$.\\
}
\end{algorithm}

\myparagraph{Insertions}
Using \sldmerge, we obtain an insertion algorithm through a simple two step process: given an input edge $e=(u,v)$, let $T_u$ and $T_v$ denote the subtrees currently containing $u$ and $v$, respectively. First, treat $e$ as a subtree, say $E$, then apply $\sldmerge(T_u,E,\allowbreak u)$, followed by $\sldmerge(T_v,T_u\cup E,v)$.
The beginning of Algorithm~\ref{alg:seq-update} shows pseudo-code for edge insertion, and the left side of Figure~\ref{fig:seq-update} shows an example of an insertion. Since the running time of \sldmerge is $O(h)$~\cite{Dhulipala2024Dendrogram}, the overall running time of the insertion algorithm is also $O(h)$.

\begin{figure*}[t]
    \centering
    \includegraphics[width=\linewidth]{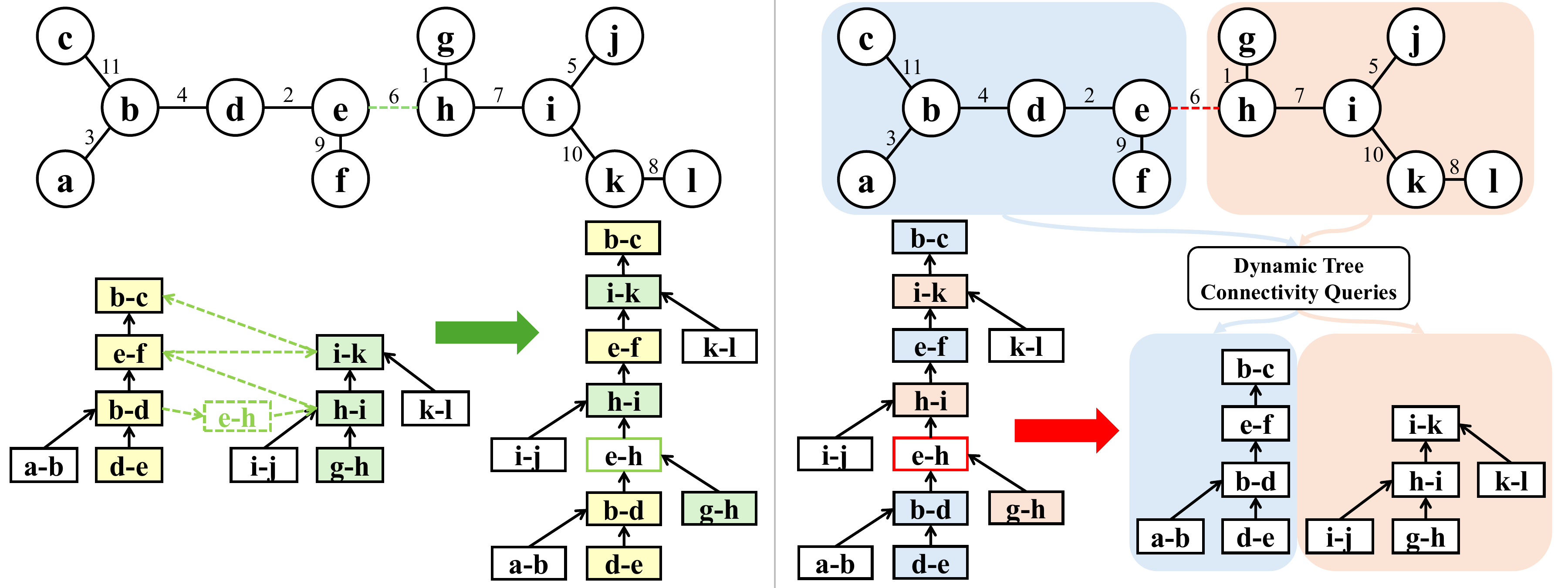}
    \caption{\small An illustration of an edge insertion and an edge deletion in \dynsld. The top of each side shows the input tree and the edge update. The bottom  of each side shows the corresponding changes to the dendrogram. The left side depicts the insertion of edge $(e,h)$ (for the sake of brevity we depict both spine merges in one step). The right side depicts its deletion. For insertions, the two characteristic spines are merged in order of increasing edge rank. For deletions, the characteristic spine is unmerged using the connectivity information from the input tree.}
    \label{fig:seq-update}
\end{figure*}

\begin{algorithm}[t]
\caption{Sequential update algorithms for dynamic single-linkage dendrograms.}\label{alg:seq-update}
\DontPrintSemicolon
\function{$\ins(u, v)$}{
    Let $e=\{u,v\}$, and let $T_u$ and $T_v$ denote the subtrees containing $u$ and $v$, respectively. Also, let $E$ denote the subtree containing only $e$.\\
    Let $T^E_u \gets T_u \cup E$. Then, $D^E_u \gets \sldmerge(T_u, E, u)$,\\
    Let $T_{uv} \gets T^E_u \cup T_v$. Then, $D_{uv} \gets \sldmerge(T^E_u, T_v, v)$.
}
\function{$\del(u,v)$}{
    Let $e=\{u,v\}$. Let $T_u$ and $T_v$ denote the subtrees containing $u$ and $v$ after deleting $e$.\\
    Let $e_u^\star$ denote the edges with minimum rank incident to vertex $u$ in $T_u$.\\
    Let $e_v^\star$ denote the edges with minimum rank incident to vertex $v$ in $T_v$.\\
    $S_u \gets \spine_{D_{uv}}(e_u^*)$ and $S_v \gets \spine_{D_{uv}}(e_v^*)$.\\
    $S'_u \gets\filter(S_u,T_u)$ and $S'_v \gets \filter(S_v,T_v)$.\\
    Update the pointers according to $S'_u$ and $S'_v$.\\
    Delete node $e$.
}
\end{algorithm}

\myparagraph{Deletions}
We now consider edge deletions. Let $e=(u,v)$ be the edge to be deleted, and let $T_u$ and $T_v$ denote the subtrees containing $u$ and $v$, respectively. The goal is to compute their respective SLDs, $D_u$ and $D_v$. 
We can view deletion as the reverse of an insertion. Specifically, let $e_u^*$ and $e_v^*$ be the minimum-rank edges incident to $u$ and $v$ in $T_u$ and $T_v$, respectively. Consider the insertion of $e$: in the two-step process described earlier, we first add $e$ into $T_u$, which appropriately merges $e$ with $spine_{D_u}(e_u^*)$, yielding the intermediate SLD $D^E_u$. Then, we merge the spines of $e_v^*$ and $e$ (in $D^E_u$) to obtain the final SLD $D_{uv}$. Importantly, the relative ordering of spines of $e_u^*$ and $e_v^*$ in $D_u$ and $D_v$ remains unchanged in the merged SLD $D_{uv}$. 

Hence, given $D_{uv}$, we construct $D_u$ by the following \emph{unmerging} process: identify the spine of $e_u^*$ in $D_{uv}$ and \emph{filter out} nodes not belonging to $T_u$. Observe that every edge except $e$ has both endpoints within the same subtree. The remaining nodes are then linked while preserving their original order, resulting in the required SLD $D_u$ ($D_v$ is computed symmetrically). Since $e$ is not present in either $T_u$ or $T_v$, it is safely removed as it is no longer referenced by any node.
To determine whether a given node belongs to $T_u$ (or $T_v$), we maintain \revision{an RC tree} on the input tree which supports connectivity queries (see Section~\ref{sec:dyntrees}). This data structure supports edge insertions and deletions in $O(\log n)$ time and allows batch connectivity queries in $O(k\log(1+n/k))$ time, where $k$ is the batch size.
As $h \geq \log_2 n$, the cost of updating the dynamic tree data structure does not change the asymptotic cost of edge insertions or deletions.
The right side of Figure~\ref{fig:seq-update} depicts an edge deletion, using the connectivity structure on the input tree to unmerge the spine.

Fetching and updating the spine requires $O(h)$ time in total. The bottleneck lies in the filtering step, which takes $O(h\log(1+n/h))$ time for $O(h)$ connectivity queries \revision{(using the RC tree over the input tree)} for the nodes in the characteristic spine. The filter operation itself runs in $O(h)$ time. Thus, the overall deletion algorithm runs in $O(h\log(1+n/h))$ time.
Theorem~\ref{thm:heightalgs} follows from this deletion algorithm and the insertion algorithm described above.
In Section~\ref{sec:tight_instances}, we show that, when parameterized by $h$, these running times are optimal---i.e., there exist input instances and sequences of insertions and deletions that require $\Omega(h)$ time per update.

\heightalgs*

\subsection{Parallel Update Algorithms}\label{sec:dynsldparallel}
The primary sequential bottleneck in the update algorithms from the previous section arises from traversing spines in the dendrogram. Both the insertion and deletion procedures assume that the SLD is stored as a linked-list-like structure, and updates are applied sequentially by traversing the characteristic spines and adjusting parent-pointers along the way.
To enable efficient parallel updates, we first need a way to extract the spine of a node in the SLD into a data structure that supports random access, such as an array. We show that this can be done in $O(\log n)$ depth and leverage this array-based representation to develop parallel update algorithms.

\myparagraph{Retrieving Spines in Low Depth}
To efficiently extract spines while ensuring efficient updates, we additionally store an RC tree (see \Cref{sec:dyntrees}) of the SLD.
Note that this is different from the RC tree introduced in the previous section which represents the {\em input tree}. The algorithms in this section maintain both an RC tree over the input tree, and an RC tree representing the SLD itself.
To compute $spine_D(e)$, we proceed as follows. First, using the RC tree of the dendrogram we identify the path decomposition for the path between the node representing $e$ and the root of the SLD. Although our RC tree is defined for an unrooted input tree, it is easy to maintain an arbitrary root of each component by storing it in the representative RC tree node.
We then ``unpack'' the clusters in the path decomposition of the spine as follows:
\begin{itemize}[itemsep=0pt,parsep=0pt,leftmargin=15pt]
    \item A base level cluster contributes a single edge to the spine.
    \item If the cluster is binary, we recursively extract spine fragments from its two children that are also binary clusters. We place the child closer to the root first. This ordering is easily maintained by the RC Tree construction and updates \revision{(i.e., maintain the invariant that the left child of an RC Tree node is the cluster with lower edge weights and vice versa)}.
\end{itemize}
Since the number of nodes visited is proportional to the spine length, the total work is $O(h)$. Computing the clusters along the path to the root requires $O(\log n)$ depth. As the spine segments corresponding to different clusters can be extracted independently, and the reconstruction proceeds by traversing downward from these clusters, the overall depth remains bounded by $O(\log n)$.
If needed, one can pack the spine nodes into a contiguous array, also in $O(\log n)$ depth and $O(h)$ work.

\myparagraph{Parallel Insertions}
For an edge insertion, we first compute the two characteristic spines using the path query described earlier. We then merge these spines in $O(h)$ work and $O(\log h)$ depth using a parallel merge operation (see Section~\ref{sec:parprelims}).
Finally, we update the parent-pointers according to the merged order, and also do a batch cut and a batch link in the RC tree representing the SLD in $O(h \log(1+n/h))$ work and $O(\log n \log h)$ depth. Overall, this algorithm requires $O(h \log (1+n/h)))$ work and $O(\log n \log h)$ depth.

\myparagraph{Parallel Deletions}
For deletions, we similarly compute the two characteristic spines and determine connectivity information for their nodes. We then apply a parallel filter operation on each spine, using the connectivity data as the predicate. Finally, we update the parent-pointers according to the sorted order of the filtered nodes, and perform a batch cut and batch link in the RC tree representing the SLD. The work is $O(h\log(1+n/h))$, and the depth is $O(\log n \log h)$.
This proves Theorem~\ref{thm:parallelalgs}:

\parallelalgs*

\subsection{Parallel Batch-Dynamic Update Algorithms}\label{sec:dynsldbatch}
We now describe a batch-dynamic algorithm for the dynamic SLD problem that processes homogeneous batches of $k$ insertions or deletions in $\tO(kh)$ work and polylogarithmic depth 

\myparagraph{Batch Insertions} Given a batch of $k$ edges to insert, we define the \emph{incidence graph} as the graph whose vertices correspond to connected components (or trees), and where two trees $U$ and $V$ are connected by an edge if there exists an input edge $(u,v)$ with $u\in U$ and $v\in V$. Notably, multiple edges in the batch may share the same endpoint in the incidence graph if they connect to the same tree, even if they do not directly share an endpoint.

\myparagraph{Handling Stars}
We first describe how to handle the case where the incidence graph forms a \emph{star graph}---i.e., all edges are incident to a single central tree $T_0$, and the remaining $k$ trees, $T_1,T_2,\ldots,T_k$, form the leaves.
As in the two-step process described before, we first merge the nodes corresponding to the new edges (in the batch) to the dendrograms of the leaves, in parallel. 
Now, if all edges in the batch shared a common vertex $v$ in $T_0$, the update becomes straightforward: we simply merge the $k$ characteristic spines from the leaf trees with the spine of $v$ since it is the common characteristic spine. However, when the $k$ merges occur at different vertices in $T_0$, they require different characteristic spines that may overlap in a non-uniform manner, making the merge more challenging. 

Let $s_i$ denote the characteristic spine of $T_i$, and let $s_i^0$ denote the corresponding characteristic spine it merges with in $T_0$. Define $D_0$ as the subtree corresponding to the union of these $k$ different characteristic spines $s_i^0$ of $T_0$. At a high level, the key idea is to split $D_0$ into ``contiguous segments'' at nodes of degree $3$. Then, partition each $s_i$ accordingly and group its parts with the corresponding segment of $D_0$ that it merges with. Finally, once all groups are formed, we merge them independently in parallel. See the right side of \Cref{fig:batch_insert} for an illustration of the star merge.

\begin{figure*}[t]
    \centering
    \includegraphics[width=\linewidth]{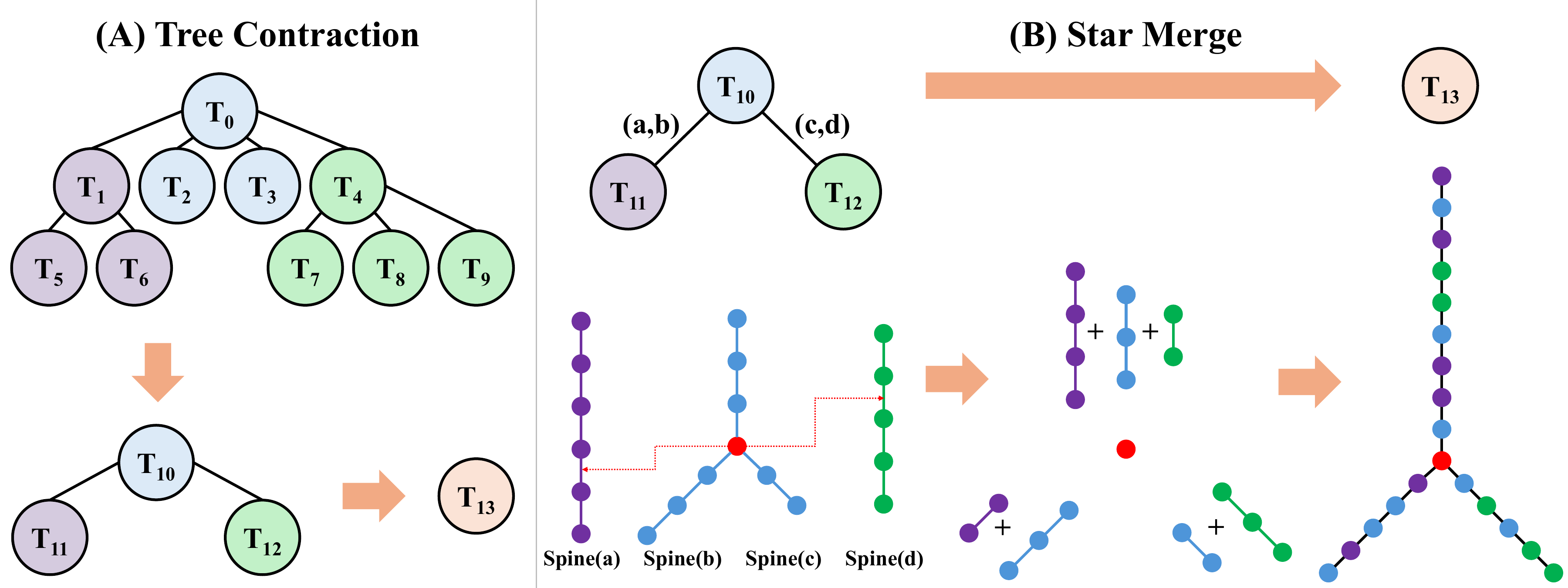}
    \caption{\small Depiction of the batch insertion algorithm. On the left, the updates form a tree where the nodes represent separate components of the input forest, and the edges represent edges in the batch of insertions between two components. The algorithm performs multiple rounds of tree contraction, where each individual contraction forms a star pattern. The right shows the process for merging the single-linkage dendrograms for multiple components in a star pattern.}
    \label{fig:batch_insert}
\end{figure*}

\begin{algorithm}[t]
\caption{Batch update algorithms for dynamic single-linkage dendrograms.}\label{alg:batch-update}
\DontPrintSemicolon
\function{$\textsc{Star-Merge}(T_0, T_1,\ldots, T_k)$}{
    Compute the characteristic spines $s_i,s_i^0$, where $s_i$ corresponds to the spine from $T_i$ and $s_i^0$ corresponds to the spine from $T_0$, for all $i$.\\
    Let $D_0 \gets $ subtree induced by $\cup_i s_i^0$.\\
    Let $B \gets $ nodes with $2$ children in $D_0$.\\
    Split $D_0$ at points in $B$ to obtain segments $d_1,d_2,\ldots$, each $d_i$ ending at either a node in $B$ or a leaf.\\
    \parfor{$i=1$ \text{to} $k$\text{ }} {
        For each $b\in B\cap s_i^0$ in parallel, identify the highest-rank node in $s_i$ whose rank is smaller than $b$.\\
        Split $s_i$ at these nodes to obtain the subspines $s_{i,1},s_{i,2},\ldots$.
        Assign $s_{ij}$ to the group corresponding to $b \in B$ that caused the split, or the leaf of $s_i^0$.
    }
    \parfor{each group $G_i$ of subspines\text{ }}{
        Merge the subspines in $G_i$ with $d_i$ to obtain a single sequence of nodes.\\
        Assign parent-pointers of nodes to their adjacent node.\\
        Assign parent-pointer of the highest-rank node to the branch point $b_i \in B$ of the segment above $d_i$.
    }
}
\function{$\textsc{Batch-}\ins(e_1,e_2,\ldots)$}{
    Compute Incidence Graph $\mathcal{I}$.\\
    \parfor{each connected component $T$ in $\mathcal{I}$}{
        \While{$|T| > 1$}{
            Apply tree contraction on $T$ and obtain a sequence of stars $\{S_1,S_2,\ldots\}$, where each $S_i = \{T_0,T_1,T_2,\ldots\}$ corresponds to one of the disjoint stars of vertices.\\
            For each $S_i$ apply $\textsc{Star-Merge}(S_i)$ in parallel.\\
            Contract each star into a single vertex, and update $T$.
        }
    }
}
\function{$\textsc{Batch-}\del(e_1,e_2,\ldots)$}{
    Apply a batch delete operation to update the RC tree used for connectivity queries.\\
    Apply $\del(e_i=(u_i,v_i))$, for all $i$, in parallel.
}
\end{algorithm}

Specifically, let $B$ denote the set of \emph{branching nodes}, which are the nodes having two children in $D_0$.
In Figure~\ref{fig:batch_insert}, the red node represents the sole branching node.
We then decompose $D_0$ by splitting it at each branching node, breaking it into a collection of paths with lower endpoints in $B$ or at a leaf. 
Next, consider the branching nodes in spine $s_i^0$. For each such branch node $b$, find the highest-rank node in $s_i$ with rank smaller than the rank of $b$, and split $s_i$ at this node.
This process decomposes $s_i$ into a collection of paths, which we call subspines. Each subspine is assigned to the segment ending at the branch node $b$ that caused the split, or the segment ending at the leaf of the corresponding spine $s_i^0$. Thus, we group the subspines from all spines $s_i$ according to the branch nodes in $B$ and leaves in $D_0$. Finally, we merge the subspines from each group in parallel to obtain the final updated SLD. We refer to this algorithm as \textsc{Star-Merge}.

\myparagraph{Work and Depth of \textsc{Star-Merge}}
We now analyze the work and depth of \textsc{Star-Merge}. Computing the $k$ spines requires $O(kh)$ work and $O(\log n)$ depth.
Given the branching points, the split-points on each spine can be computed in $O(h)$ work and $O(\log h)$ depth using a parallel merge primitive.
Merging a group of subspines with a total of $x$ nodes takes $O(x)$ work and $O(\log k \log h)$ depth using a parallel merge primitive. Since the total number of nodes across all groups is $O(kh)$, the merging step requires $O(kh)$ work and $O(\log k \log h)$ depth. The bottleneck, however, is updating the RC tree representing the SLD. Since there are at most $O(kh)$ updates, the RC tree is updated in $O(kh\log(1+n/(kh)))$ work and $O(\log n \log(kh))$ depth. Thus, overall, \textsc{Star-Merge} requires $O(kh \log(1+n/(kh)))$ work and $O(\log n \log(kh))$ depth.

\myparagraph{The General Case}
Finally, consider the general case of $k$ batch insertions where the incidence graph forms a forest. We solve this with the help of \defn{tree contraction}~\cite{miller1985parallel}. Given a tree on $n$ vertices, tree contractions iteratively finds a maximal independent set of degree 1 and degree 2 vertices, and merges them into one of its neighbors.
This process is repeated until a single vertex remains. The number of rounds of the tree contraction is bounded by $O(\log n)$.
Crucially, the sets of vertices that are merged in each round are disjoint and form stars (one or more vertices merging into a single center vertex).
The idea is to repeatedly compute a single round of maximal tree contractions on the incidence graph, and apply the \textsc{Star-Merge} algorithm independently for each star contraction.
The full method is depicted in Figure~\ref{fig:batch_insert}.

The total work for tree contraction is $O(k)$ and it takes $O(\log k)$ depth to compute the independent set of contractions each round~\cite{anderson2024deterministic, anderson2023thesis}.
The overall work during all star merges remains $O(kh \log(1+n/(kh))$, which dominates the total work. The depth is now $O(\log n\allowbreak\log k\log(kh))$ where the additional $\log k$ factor is due to the number of rounds of tree contraction.

\myparagraph{Batch Deletions} Batch deletions are comparatively easier to handle. After updating the connectivity data structure, we concurrently run the parallel spine unmerge operation (previously described for single edge deletion) on each deleted edge.
Although some of the spines may overlap, multiple parent changes to any node all result in the same value, and thus it does not affect the correctness of the algorithm (although some of the work may be redundant).

We will now analyze the work and depth of the batch deletion algorithm.
Updating the connectivity structure takes $O(k \log (1+n/k))$ work and $O(\log n \log k)$ depth.
Computing the $O(k)$ spines once again requires $O(kh)$ work and $O(\log n)$ depth.
The connectivity information can be obtained by a batch query in $O(kh\log(1+n/(kh)))$ work and $O(\log n)$ depth.
Each spine unmerge happens independently in parallel. In total this requires $O(kh)$ work and $O(\log h)$ depth.
Finally, updating the RC tree takes $O(kh\log(1+n/(kh)))$ work and $O(\log n \log(kh))$ depth.
Overall, the batch deletion takes $O(kh\log(1+n/(kh)))$ work and $O(\log n \log(kh))$ depth.
Theorem~\ref{thm:batchalgs} follows from this batch-deletion algorithm and the previous batch-insertion algorithm:

\batchalgs*

\section{Output-Sensitive Algorithms}\label{sec:outputsensitive}

\begin{figure*}[t]
    \centering
    \includegraphics[width=0.8\linewidth]{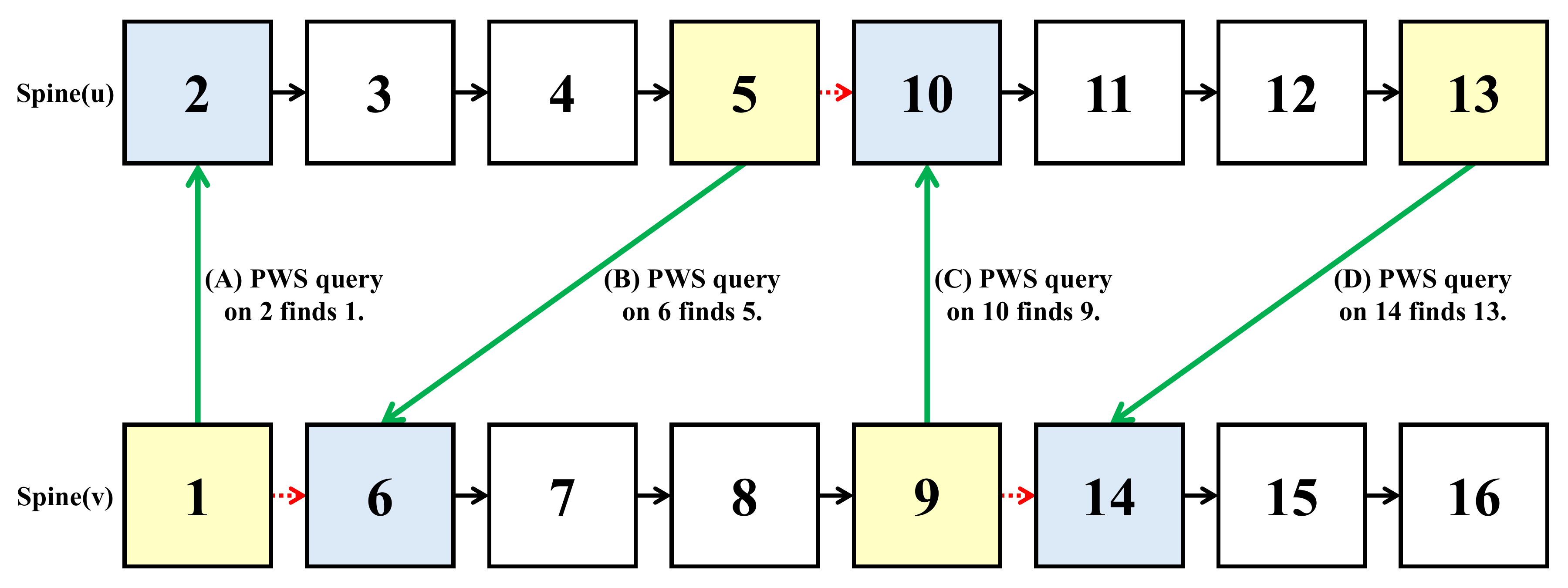}
    \caption{\small A depiction of the output-sensitive algorithm for merging two spines. The nodes are depicted with their weight rather than the edge they represent. In each step, a blue node uses a PWS query to find the node in the opposite spine which should be its child. The child is depicted in yellow and changes its parent-pointer. The algorithm then proceeds with the previous parent of the yellow node.}
    \label{fig:sld_c_merge}
\end{figure*}

In this section, we describe algorithms for edge insertions in a dynamic SLD with cost proportional to the number of changes in the dendrogram. We call these algorithms {\em output-sensitive}, as the cost depends on how much the ``output'' dendrogram changes.
More formally, let $c$ denote the number of parent-pointer changes in the output SLD as a consequence of an edge insertion. Our goal is to design $\Tilde{O}(c)$ cost dynamic update algorithms.

\subsection{New Dynamic Tree Queries}
As a prerequisite for our output-sensitive dynamic SLD algorithms, we first describe two useful queries on trees that can be answered efficiently using RC trees. Although both of these queries involve information about paths in the tree, they do not directly fit into the general framework for path queries established by prior work on RC trees. Thus, we also provide a detailed description of their implementations.

\begin{definition}[Path Weight Search Query]\label{def:pws}
    Given a vertex-weighted tree $T$, a path weight search query takes vertices $u$ and $v$ in $T$ such that the weights on the path from $u$ to $v$ are increasing, and a weight $w$. It returns the maximum-weight vertex on the path whose weight is less than $w$ (and/or the minimum-weight vertex on the path whose weight is greater than $w$).
\end{definition}

To implement path weight search (PWS) queries in RC trees, each binary cluster must store the minimum and maximum weight of any vertex on its cluster path. These values can be easily maintained during RC tree links and cuts. We refer to prior work on RC trees~\cite{acar2004dynamizing, acar2005experimental, anderson2023thesis} for the maintenance of such aggregate values over the cluster paths in each binary cluster.

The query begins by computing the path decomposition of the path between $u$ and $v$, which gives $O(\log n)$ clusters whose cluster paths are disjoint and whose union forms the entire path from $u$ to $v$. Since the path is strictly increasing in weight, these binary clusters define an ordered sequence of disjoint {\em weight ranges}.
If $w$ falls between two of these weight ranges, the largest value in the lower range is returned.
Otherwise, the search proceeds into the cluster whose weight range contains $w$.
The query recursively descends into the child cluster whose weight range contains $w$ until $w$ falls between two ranges or a cluster containing a single vertex is reached.
This \revision{query takes time proportional to the height of the RC tree, which is} $O(\log n)$.

\begin{definition}[Path Median Query]
    Given a tree $T$, a path median query takes two vertices $u$ and $v$ in $T$ and returns the vertex at index $\lfloor \ell/2 \rfloor$ (referred to as the median vertex) in the sequence of vertices on the path between $u$ and $v$, where $\ell$ is the length of the path.
\end{definition}

To implement path median queries in RC trees, every binary cluster stores the length (number of vertices) of its cluster path, which, as before, is easy to maintain during updates.
Similar to PWS queries, the query begins by computing the path decomposition of the path between $u$ and $v$. The total path length $\ell$ is then determined by summing the lengths of these cluster paths lengths (plus the number of clusters in the path decomposition plus one, since vertices ``between'' two clusters are not considered as part of the cluster path by definition).
If the vertex corresponding to index $\lfloor \ell/2 \rfloor$ (the target median) falls on a vertex that connects two of the clusters in the path decomposition, this vertex is returned by the query. Otherwise, the search recursively proceeds into the binary cluster that contains the target index vertex. If this cluster's path contains the vertices with indices in the range $[l,r]$, then we should search for the vertex at index $\lfloor \ell/2 \rfloor - l$ in the cluster path of that cluster.
Further, this binary cluster has two children that are binary clusters, and are both incident to a common vertex. If the common vertex is the target index vertex, the recursion ends. Otherwise, we recurse into one of the two children appropriately.
\revision{Similar to PWS queries, this query takes time proportional to the height of the RC tree, which is} $O(\log n)$.

\subsection{Incremental Output-Sensitive Algorithm}

We now describe an incremental output-sensitive algorithm for a single edge insertion $e=(u,v)$ in the dynamic SLD problem. As before, the algorithm maintains an RC tree over the SLD.
Let $e_u^*$ and $e_v^*$ be the minimum-rank edges incident to $u$ and $v$ in $T_u$ and $T_v$, respectively.
In the original algorithm, we first merge $\spine(e_u^*)$ with the new edge node $e$, and then merge $\spine(e_v^*)$ with $\spine(e)$.
We now describe how to merge two spines in $O(c \log(1+n/c))$ time, where $c$ is the number of parent-pointer changes caused by the merge.
For the first merge with edge node $e$, $c=O(1)$, so the cost is $O(\log n)$; even if these parent-pointer changes are undone in the second merge, the total cost of the insertion is $O(\log n) + O(c \log(1+n/c)) = O(c \log(1+n/c))$, where $c$ is the total number of parent-pointer changes between the original and new dendrograms.

\myparagraph{Output-Sensitive Spine Merging}
We now describe the new algorithm for merging two spines $\spine(u)$ and $\spine(v)$.
The algorithm first discovers all required parent-pointer changes and then uses a batch cut operation followed by a batch link operation in the RC tree to update the pointers.
Without loss of generality, let the weight of $u$ be greater than that of $v$.
This means that there exists a node $x_v$ in $\spine(v)$ that must have its parent changed to $u$. To find $x_v$, we will use a PWS query on $\spine(v)$ using the weight of $u$. Then we set the parent of $x_v$ to be $u$. If $x_v$ was the last node of $\spine(v)$, the algorithm is complete. Else, $x_v$ had a parent $p_v$ and there is a node $x_u$ in $\spine(u)$ that must have its parent changed to $p_v$. Once again, we can determine $x_u$ using a PWS query on $\spine(u)$ with the weight of $p_v$.

The algorithm continues this alternating pattern of finding the node in the other spine that should now point to the node in the current spine, until the found node doesn't have a parent (i.e., it was the last node in one of the spines).
Figure~\ref{fig:sld_c_merge} illustrates this process.
Clearly, the number of PWS queries is exactly the number of parent-pointer changes. Thus, the time for this part of the algorithm is $O(c \log n)$.
Performing batch cut and batch link operations with $c$ edges in the RC tree takes $O(c \log(1+n/c))$ time, and thus, we can immediately upper bound our update cost with $O(c \log n)$.

\begin{figure*}[t]
    \centering
    \includegraphics[width=0.75\linewidth]{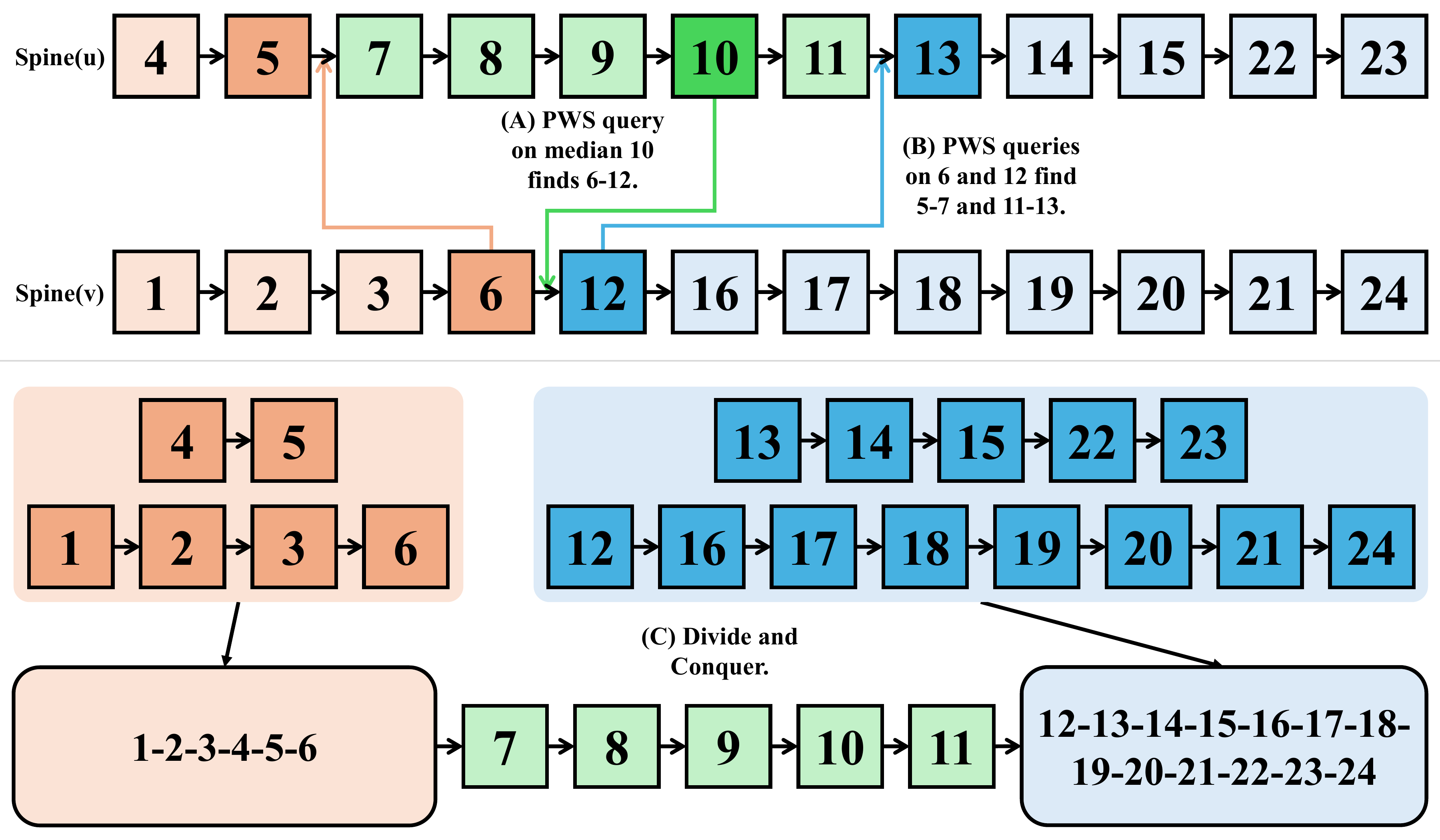}
    \caption{\small A depiction of the parallel output-sensitive spine merge algorithm. First the median value (dark green node) in one spine is found. Then, PWS queries are used to find the values $x_v$ and $y_v$ (dark orange and dark blue nodes) in the other spine ``surrounding'' that median value. Next, PWS queries are used to separate the original spine at split points $x_u$ and $x_v$ (dark orange and dark blue nodes) to the left and right of the median. Finally, both spines are split, the pairs of lower and upper spines are merged recursively in parallel, and the results are combined.}
    \label{fig:sld_par_c_merge}
\end{figure*}

\myparagraph{Better Cost Bound}
Next we describe how to perform the $c$ PWS queries in $O(c \log(1+n/c))$ total time, thereby achieving this same improved bound for the total cost of the insertion.
First, the initial part of each PWS query to get the path decomposition in the RC tree can be done once for both spines and then stored in order by increasing weight ranges.
The remainder of each query involves traversing down the RC tree starting from the path decomposition clusters. Intuitively, these downward query paths may overlap a lot, especially at higher clusters in the RC tree.
We also observe that the PSW queries in each spine occur in increasing order of weight.
To take advantage of these properties, each query (other than the first in each spine) will proceed differently from a normal query. The query will start at the cluster at which the last query finished, i.e., the cluster in which the downward recursion terminated and a vertex was returned. It then traverses up the RC tree until it either reaches one of the path decomposition clusters, or it finds a cluster with a weight range containing the new query weight. In the first case, it will start its search from the cluster with the next larger weight range, and search downward normally. In the latter case, it can start the normal downward search immediately.

This algorithm works because for any weight range encountered during the upward traversal which does not contain the new query weight, the new query weight (and all future query weights) are larger than this entire range. Thus, these clusters will not need to be visited in any future query, since the queries happen in increasing order of query weight.
As a result, it is easy to see that each node will be visited at most twice.
Thus, the total time is asymptotically the same as the number of nodes in the union of the $c$ paths traversed.
Prior work~\cite{acar2020changeprop, anderson2023thesis, anderson2024deterministic} has shown that the number of nodes in any $c$ paths towards the root in an RC tree is $O(c \log(1+n/c))$.
Overall, this leads to an output-sensitive insertion algorithm that takes $O(c \log(1+n/c))$ time, proving Theorem~\ref{thm:changealgs}.

\changealgs*

\subsection{Parallel Output-Sensitive Algorithm}

We now describe a parallel output-sensitive algorithm for merging two dendrogram spines.
Our approach combines median queries, PWS queries, and a divide-and-conquer process over the height of one spine. Figure~\ref{fig:sld_par_c_merge} illustrates our algorithm. The high-level approach is as follows:

\begin{enumerate}[itemsep=0pt,parsep=0pt,leftmargin=15pt]
    \item Use a path median query to find a ``central'' vertex in $\spine(u)$.
    \item Use a PWS query to find points $x_v$ and $y_v$ such that the median falls between these two points in $\spine(v)$.
    \item Use two additional PWS queries to find points $x_u$ and $y_u$ in $\spine(u)$ that immediately precede $x_v$ and immediately succeed $y_v$, respectively.
    \item Divide both spines into a lower portion and an upper portion based on the four points determined in the previous steps.
    \item Recursively merge the two lower portions of the spines and the two upper portions of the spines in parallel.
    \item Concatenate the merged lower spine and merged upper spine.
\end{enumerate}

Our algorithm performs divide-and-conquer on $\spine(u)$ (or the smaller spine for extra efficiency).
First, we use a median query to find the weight $w$ of a central node in $\spine(u)$. Then, we use a PWS query to find in $\spine(v)$ the largest-weight node with weight less than $w$ and the smallest-weight node with weight greater than $w$; call these $x_v$ and $y_v$, respectively.
Next, we use two additional PWS queries to find in $\spine(u)$ the largest-weight node with weight smaller than the weight of $x_v$, and the smallest-weight node with weight larger than the weight of $y_v$; call these $x_u$ and $y_u$, respectively.
Then, we recursively, and in parallel, merge the spine from $x_u$ down with the spine from $x_v$ down, and the spine from $y_u$ up with the spine from $y_v$ up.
This results in a merged lower spine and a merged upper spine. Finally, we concatenate the merged lower spine with the sub-spine of $\spine(u)$ between (but not including) $x_u$ and $y_u$, and then the merged upper spine.

Since $x_u$ is below the median of $\spine(u)$ and $y_u$ is above the median, the upper and lower spines of $\spine(u)$ both have length at most half as big as the original spine. Thus, this recursion will end in $O(\log h)$ depth. Each divide step takes $O(\log n)$ depth, and combining the merged spines takes $O(1)$ depth, so the overall depth is $O(\log n \log h)$.
This proves the depth bound in Theorem~\ref{thm:parallelchangealg}.
Each divide step uses one median query and 3 PWS queries which is $O(\log n)$ work. The combine steps all take $O(1)$ work.
The $O(\log n)$ work in each divide step can be charged to the parent-pointer change of $x_v$, and the previous child of $y_u$, so the work is $O(c \log n)$.

\revision{
\noindent\textbf{Remark:}}
A careful reader may wonder why we cannot just split $\spine(u)$ before or after the median, and instead have to use two additional PWS queries to find $x_u$ and $y_u$.
This is because, if we just split at the median, we cannot charge these queries to a specific parent-pointer change. The algorithm must carefully split the spines only at points that should have a changed parent-pointer in the final merged spine, so that the work of the split can be charged to these parent-pointer changes.
Otherwise, the total number of recursive steps could be as large as $2^{\log h} = \Omega(h)$.
A good algorithm must carefully charge the work while ensuring that the divided spines are both at most half as large as the original spine to ensure the recursive depth stays low.
We accomplish this by finding two split points to the left and right of the median.

\revision{

\myparagraph{Better Work Bound}
To achieve work efficiency with respect to the sequential output-sensitive insertion algorithm, we would like to achieve a bound of $O(c \log (1+n/c))$ for the work of our parallel algorithm. Here, we show how to accomplish this, thereby proving the work bound from Theorem~\ref{thm:parallelchangealg}.

We introduce two key modifications to the algorithm. First, we alternate between $\spine(u)$ and $\spine(v)$ as the spine from which we initiate the median query at each recursive level. This ensures that every two levels, both spines decrease in length by at least a factor of $2$, not just $\spine(u)$.
Second, we perform RC tree links and cuts immediately upon splitting or joining a spine, rather than batching them until the end of the update. This ensures that the cost of the PWS and the median queries (as well as the link and cut operations) scales with the spine's size.

We will now prove that this algorithm takes $O(c \log (1+n/c))$ work.
Consider the {\em binary recursion tree} of the divide-and-conquer algorithm, which has $O(c)$ nodes. In this rescursion tree, for a node $x$ at recursive depth $d(x)$ (i.e., the recursive depth of the subproblem $x$), both spines have length $O(n/2^{d(x)/2})$. Thus, the cost of the median and PWS queries (as well as the RC tree links and cuts) in subproblem $x$ is $O(\log (n/2^{d(x)/2})) = O(\log n - d(x))$.
Now, imagine transforming this binary recursion tree into a ternary tree as follows: for each node $x$, add a new child which is a path of length equal to the cost of the queries in subproblem $x$. Since each node $x$ is at depth $d(x)$ and the new path has length $O(\log n - d(x))$, the height of this ternary tree will be $O(\log n)$. As the recursion tree initially had $O(c)$ nodes and we added one path for each node, the ternary tree has at most $O(c)$ leaves.

The total cost across all queries and updates (and the algorithm's overall work) are proportional to the number of nodes in this ternary tree, which represents the union of the $O(c)$ paths from leaves to the root.
At depths of $O(\log c)$ or lower, there are at most $O(c)$ total nodes. For depths greater than $O(\log c)$, each path has a length of at most $O(\log n - \log c) = O(\log (1+n/c))$. In the worst-case, where all these paths are disjoint, the total number of nodes is $O(c \log (1+ n/c))$. Thus, the total number of nodes, and the total work of the algorithm, is $O(c \log (1+n/c))$.

}

\parallelchangealg*
\revision{

\section{Discussion of Optimality}\label{sec:tight_instances}

In this section, we describe various notions of optimality for the fully-dynamic single-linkage dendrogram problem, establish lower bounds, and assess the optimality of the algorithms presented in this paper.
Consider the following theorem.

}

\begin{theorem}\label{thm:update_lb}
    There exists an input graph and a corresponding update that affects $\Omega(h)$ pointers in the output SLD.
\end{theorem}
\begin{proof}
    Consider the graph that comprises $n/(h+1)$ disjoint stars (say $T_1, T_2,\ldots, T_{n/(h+1)}$), each containing $h+1$ vertices, such that the weights of edges in $T_1$ are $(1,h+1,2h+1,\ldots)$, in $T_2$ are $(2,h+2, 2h+2,\ldots)$, and so on. Their respective SLDs (say $D_1, D_2, \ldots, D_{n/(h+1)}$) will each be path graphs of height $h$. 

    Now, consider an edge connecting the centers of stars $T_i$ and $T_j$ (arbitrary $i$ and $j$). Let the weight of the new edge be $0$ (anything smaller than the weights of every edge in the graph). The new SLD will be a path graph as well of height $2h+1$, with every node in $D_i$ pointing to the first edge node with weight greater than it, and similar for nodes in $D_j$ as well. Thus, the number of nodes with affected pointers are $2h+1$. 

    Similarly, if we apply a deletion to the edge of weight $0$, we get back the original SLDs, thus affecting $2h+1$ nodes again (with the node corresponding to weight $0$ now deleted). 
\end{proof}

\revision{

Theorem~\ref{thm:update_lb} establishes an $\Omega(h)$ lower bound on the worst-case update cost for the dynamic SLD problem, as it requires explicitly maintaining the dendrogram.
Furthermore, since instances with $h = \Omega(n)$ exist, this implies an $\Omega(n)$ lower bound in the worst case.
All of the algorithms presented in this paper take $O(n)$ time/work in such cases, making them optimal with respect to $n$.
They are also near-optimal with respect to $h$, as they require $\Tilde{O}(h)$ time/work.

We also consider a finer notion of optimality: the lower bound for a specific update on a specific instance, rather than a worst-case across all instances and updates.
In this paper, the parameter $c$ precisely quantifies this measure, denoting the number of structural changes (i.e., parent-pointer changes) to the dendrogram caused by a given update. Since the dynamic SLD problem requires explicitly maintaining the dendrogram, there is a natural tight lower bound of $\Omega(c)$ for a specific update on a specific instance. Our output-sensitive insertion algorithm achieves $\Tilde{O}(c)$ time/work, making it near-optimal under this definition.
}

\section{Dynamic SLD Applications}\label{sec:applications}
In this section, we describe some of the possible queries one can perform on an SLD, and we contrast this with just maintaining a minimum spanning forest (MSF) of the graph. We also describe a major related problem of the dynamic SLD problem, Cartesian Trees, and show how our dynamic algorithms can be applied to also provide dynamic algorithms for this problem.
In Table~\ref{tab:applications} we summarize the costs of the various queries described in this section. We list the work and depth of each operation on an explicit dendrogram with \dynsld, and on just an MSF of the graph.

\subsection{Dendrogram Queries}
In this section, we study queries on the SLD that provide deeper insight into the structure of the clustering. Beyond simply maintaining the dendrogram, efficient queries are crucial for analyzing and interacting with the clustering in real time. For instance, retrieving cluster memberships, lowest common ancestors (LCA), or subtree queries allows for fast anomaly detection, de-duplication, and interactive exploration of hierarchical patterns. Such queries are especially useful in the dynamic setting, where clusters evolve over time and efficient updates are essential to maintain interpretability and usability of the dendrogram. We now discuss some of these useful queries in detail. 

\myparagraph{Threshold or LCA Queries} \emph{``Given query vertices $s,t$ and query threshold $\tau$, are $s$ and $t$ in the same cluster if we agglomeratively cluster until all distances are above threshold $\tau$?''} This query essentially reduces to an LCA query on the SLD. However, as noted in \cite{tseng2022parallel}, it can be computed more directly via a path query on the input: find the maximum-weight edge on the path between $s$ and $t$. We can compare the weight of this edge with $\tau$ to answer the query. The maximum-weight edge can be obtained in $O(\log n)$ time using a standard path query in an RC tree (with augmented values) on the input, which \dynsld already maintains for connectivity queries.

\myparagraph{Cluster Size Queries} \emph{``Given a query vertex $u$ and query threshold $\tau$, what is the size of the cluster containing $u$ if we agglomeratively cluster until all distances are above threshold $\tau$?''} This query reduces to finding the size of the subtree rooted at the node $e$ with the largest weight smaller than $\tau$ on the spine of $u$ in the SLD, i.e. the path between the lowest weight edge incident to $u$ and the root. We compute this with a path query followed by a subtree query on the RC tree representing the SLD. Specifically, $e$ can be computed using the path weight search (PWS) query (see \Cref{def:pws}) in $O(\log n)$ time. The size of the subtree rooted at $e$ can be obtained in $O(\log n)$ time using a standard subtree query of an RC tree (with augmented values), which is again already maintained by \dynsld.

\myparagraph{Cluster Report Queries} \emph{``Given a query vertex $u$ and query threshold $\tau$, return the cluster containing $u$ when we agglomeratively cluster until all distances are above threshold $\tau$.''} This query corresponds to returning the entire subtree $S$ rooted at node $e$, as defined in the previous query. This is done as follows: first, using the RC tree of the SLD we identify the path decomposition for the path between the leaf $u$ (i.e., the node representing the lowest weight edge incident to $u$) and the node representing $e$. We then ``unpack'' the clusters in the path decomposition by traversing down from each cluster and collect all the base level edges, except cluster $e$. For cluster $e$, we traverse down, but ignore child clusters that are closer to the root. Return the set $S$ corresponding to the union of the endpoints of all the base level edges computed. 
Since the number of nodes visisted is proportional to the size of $S$, the total work is $O(|S|)$. Computing the path decomposition takes $O(\log n)$ depth, followed by unpacking the clusters requires traversing down the RC tree which also takes $O(\log n)$ depth, leading to an overall depth of $O(\log n)$.

We can also obtain a flat clustering---i.e., return all clusters formed when we agglomeratively cluster until all distances are above threshold $\tau$---in $O(n)$ work and $O(\log n)$ depth by a simple extension of the idea discussed above. This operation can be supported by the dynamic MSF algorithm as well. However, it is not clear whether it can support the cluster reporting operation in parallel with output-sensitive work.

\noindent\textbf{Remark:} In particular, we would like to note that any query expressible as a combination of path and/or subtree queries on the SLD and/or the input can be efficiently computed in parallel.

\begin{table}[t]
    \centering
    \begin{tabular}{|c|c c|c c|}
        \hline
        \multirow{2}{*}{Operation} & \multicolumn{2}{c|}{\dynsld} &
         \multicolumn{2}{c|}{Dynamic MSF} \\
         & Work & Depth & Work & Depth \\
         \hline\hline
         Threshold & \multirow{2}{*}{$O(\log n)$} & \multirow{2}{*}{$O(\log n)$} & \multirow{2}{*}{$O(\log n)$} & \multirow{2}{*}{$O(\log n)$} \\
         Query & & & & \\
         Cluster Report & \multirow{2}{*}{$O(\revision{|S|})$} & \multirow{2}{*}{\boldmath $O(\log n)$} & \multirow{2}{*}{$O(\revision{|S|})$} & \multirow{2}{*}{$O(\revision{|S|})$} \\
         Query & & & & \\
         Cluster Size & \multirow{2}{*}{\boldmath $O(\log n)$} & \multirow{2}{*}{\boldmath $O(\log n)$} & \multirow{2}{*}{$O(\revision{|S|})$} & \multirow{2}{*}{$O(\revision{|S|})$} \\
         Query & & & & \\
         \hline
    \end{tabular}
    \caption{\small{} The work and depth bounds for various SLD query operations using \dynsld versus using just a dynamic minimum spanning forest. \revision{Here, \boldmath$S$ is the set of points contained in the cluster of interest.}}
    \label{tab:applications}
\end{table}

\subsection{Cartesian Trees}
Given an array $A$ with $n$ elements, the \defn{cartesian tree} is a binary tree representing $A$ that is constructed as follows: 
\begin{itemize}[itemsep=0pt,parsep=0pt,leftmargin=15pt]
    \item the root of the tree is $A[i]=\min\{A[1],\ldots A[n]\}$,
    \item the left subtree is recursively computed on $A[1\ldots i-1]$ and the right subtree is  recursively computed on $A[i+1\ldots n]$.
\end{itemize}
Cartesian trees are a fundamental data structure with applications to range minimum queries (RMQ), all near small values (ANSV), suffix arrays, bottleneck edge queries (BEQ), etc.~\cite{vuillemin1980unifying, bender2000lca, demaine2014oncartesian,SB14cartesian}. 
The concept of cartesian trees can also be extended to tree inputs, which was studied by \cite{demaine2014oncartesian}.

Cartesian trees satisfy (i) heap order on the nodes, and (ii) the in-order traversal of the tree corresponds to $A$. 
In \cite{Dhulipala2024Dendrogram}, the authors observed that the cartesian tree of an array (or tree) is equivalent to the single-linkage dendrogram of the path graph (or tree) with edge weights corresponding to the array entries (or node weights), which is evident from the properties mentioned above. We leverage this relationship between the problems to design dynamic algorithms for maintaining cartesian trees under updates. 
Henceforth, we may assume max-heap order on the nodes for convenience.
To avoid confusion, we assume that the input is in tree format, where elements (and their values) correspond to edges (and their weights).
We consider two types of updates:
\begin{enumerate}[itemsep=0pt,parsep=0pt,leftmargin=15pt]
    \item \textbf{Leaf Updates or Appends:} Here, we assume that the updates only insert/delete leaves. Using our output-sensitive algorithm from \Cref{sec:outputsensitive}, we get a $O(\log n)$ time algorithm for insertions, since $c=O(1)$. For deletions, observe that the element being deleted can have at most one child in the cartesian tree, thus requiring $O(1)$ changes as well--i.e., its child must now point to its grand-parent. The child can be found in $O(\log n)$ time using a PWS query, followed by updating its parent, leading to an overall deletion time of $O(\log n)$.
    \item \textbf{Arbitrary Updates:} Insertions correspond to splitting a vertex $u$ and introducing a new edge $e=(u,u')$; this is equivalent to inserting an element at an arbitrary index in the case of array inputs. Let $v$ be some neighbor of $u$. Then, this update can be processed as: delete $e'=(u,v)$, insert $e=(u,u')$ and finally insert $e''=(u',v)$ where weight of $e''$ is the weight of $e'$. 
    Deletions correspond to contracting the edge $e=(u,v)$; this is equivalent to deleting an element at an arbitrary index in the case of array inputs. Let $w \neq v$ be some neighbor of $u$. Then, this update can be processed as: delete $e=(u,v)$, delete $e'=(u,w)$ and finally insert $e''=(v,w)$ where weight of $e''$ is the weight of $e'$. 
    The overall update times are identical to \dynsld update times.
\end{enumerate}
We note that for leaf inserts and deletes, Demaine et al.~\cite{demaine2014oncartesian} gave amortized $O(\log n)$ time algorithms, whereas we give worst case $O(\log n)$ time bounds.






\section{Related Work}

\myparagraph{Static SLD Computation}
Sequentially, an SLD can be computed using Kruskal's algorithm in $O(n\log n)$ time.
The work of Demaine et al.~\cite{demaine2014oncartesian} gave an algorithm for SLD computation that takes $O(n \log n)$ time and showed that if the cost of sorting the edges is ``free'', SLD can be solved in $O(n)$ time.
Wang et al.~\cite{wang2021fast} gave the first work-efficient parallel algorithm for SLD computation. Their algorithm is randomized and computes the SLD in $O(n \log n)$ expected work and $O(\log^2 n \log\log n)$ depth with high probability.
More recently, Dhulipala et al.~\cite{Dhulipala2024Dendrogram} gave an algorithm that takes $O(n \log h)$ work and $O(\log^2 n \log^2 h)$ depth, where $h$ is the height of the dendrogram. They further showed a lower bound proving that this work is optimal.

\myparagraph{Dynamic HAC Heuristics}
Very recently, Yu et al.~\cite{yu2025dynhac} designed a practical algorithm for solving dynamic average-linkage graph-based HAC under (batch) edge updates in the fully-dynamic setting.
Their approach is based on a recent distributed algorithm for solving average-linkage graph-based HAC~\cite{dhulipala2023terahac}.
An earlier work by Monath et al.~\cite{monath2023online} give an incremental heuristic algorithm for dynamic average-linkage graph-based HAC based on merging and unmerging clusters. Only the algorithm by Yu et al.~\cite{yu2025dynhac} achieves a provable approximation ratio for graph-based HAC in the dynamic setting, although in the worst case it can resort to full recomputation of the dendrogram. We note that earlier work by Tseng et al.~\cite{tseng2022parallel} give conditional lower bounds showing that maintaining dendrograms for average-linkage and other popular linkage functions in the fully dynamic setting requires polynomial update time.

\myparagraph{Dynamic Minimum Spanning Forest}
The most well-known algorithm for fully-dynamic minimum spanning forest is that of Holm et al.~\cite{holm2001poly} which supports updates in $O(\log^4 n)$ amortized time.
\revision{Combining this algorithm with the dynamic SLD algorithms in this paper would yield a solution to the fully-dynamic single-linkage clustering problem which processes updates in $O(\log^4 n + h \log (1+n/h))$ amortized time, where $h$ is the height of the dendrogram. As long as $h$ is $O(\log^3 n)$, the cost is dominated by the dynamic MSF algorithm, with $O(\log^4 n)$ amortized work per update.}

The first fully-dynamic batch-parallel algorithm for minimum spanning forest was given by Tseng et al.~\cite{tseng2022parallel} which processes batches of $k$ updates with $O(k \log^6 n)$ expected amortized work and $O(\log^4 n)$ depth with high probability.
\revision{Combining this algorithm with the batch-dynamic SLD algorithm in this paper, batches of updates can be processed in $O(k \log^6 n + kh \log (1+n/(kh)))$ expected amortized work, and $O(\log^4 n)$ depth with high probability.}

\myparagraph{Dynamic Cartesian Trees}
A prior study~\cite{birula2006amortized} showed how to dynamically maintain cartesian trees under insertions and ``weak'' deletions in $O(\log n)$ amortized time (weak deletions are a type of lazy deletion where nodes are marked as deleted and periodically removed from the tree).
Demaine et al.~\cite{demaine2014oncartesian} gave algorithms for leaf insertions and deletions that run in $O(\log n)$ amortized time.
Our results in this paper improve over the results of Demaine et al.~\cite{demaine2014oncartesian} by giving algorithms for leaf insertions and deletions that run in $O(\log n)$ worst case time, and our results are the first to extend to arbitrary parallel batch-dynamic insertions and deletions.
\section{Conclusion}
In this paper, we gave efficient sequential and parallel update algorithms for dynamic and batch-dynamic single-linkage dendrogram computation on a weighted input tree. All of our algorithms perform asymptotically less work than optimal algorithms for static recomputation of the dendrogram~\cite{Dhulipala2024Dendrogram}, and perform much less work in many cases.
In combination with existing dynamic minimum spanning forest algorithms~\cite{holm2001poly, tseng2022parallel}, our methods complete the theoretical picture for an end-to-end dynamic single-linkage dendrogram algorithm given an input point set or graph.

In this paper, we also gave a nearly-optimal output-sensitive edge insertion algorithm that requires work proportional to the number of structural changes to the dendrogram caused by the update. 
An open question here is whether it is possible to design similar output-sensitive algorithms for edge deletions, \revision{thereby obtaining fully-dynamic output-sensitive update algorithms. It would also be interesting to extend our techniques in the batch-dynamic setting for height-bounded updates to obtain output-sensitive updates.}
Furthermore, it would be interesting to study lower bounds for dendrogram maintenance algorithms in, for example, the cell-probe model~\cite{cellprobe} (which has been used to prove lower bounds for other dynamic graph problems~\cite{patrascu2004lower}), to prove tighter optimality of our algorithms, height-bounded and output-sensitive algorithms.

An other interesting question is whether we can develop practical implementations for dynamic single-linkage dendrogram maintenance.
Our algorithms are a good candidate for this because they have good theoretical guarantees, and our parallel algorithms are in the realistic binary fork-join model~\cite{blelloch2019optimal} which closely models modern multicore hardware.
An other interesting question for future work is to try to dynamize the static RC tree tracing (RCTT) algorithm~\cite{Dhulipala2024Dendrogram} which requires $O(n\log n)$ work but was shown to perform extremely well in practice.

\section*{Acknowledgments}
This work is supported by NSF grants CCF-2403235 and CNS-2317194. We
thank the anonymous reviewers for their helpful comments.

\bibliographystyle{ACM-Reference-Format}
\bibliography{main.bib}


\begin{thebibliography}{59}


\ifx \showCODEN    \undefined \def \showCODEN     #1{\unskip}     \fi
\ifx \showISBNx    \undefined \def \showISBNx     #1{\unskip}     \fi
\ifx \showISBNxiii \undefined \def \showISBNxiii  #1{\unskip}     \fi
\ifx \showISSN     \undefined \def \showISSN      #1{\unskip}     \fi
\ifx \showLCCN     \undefined \def \showLCCN      #1{\unskip}     \fi
\ifx \shownote     \undefined \def \shownote      #1{#1}          \fi
\ifx \showarticletitle \undefined \def \showarticletitle #1{#1}   \fi
\ifx \showURL      \undefined \def \showURL       {\relax}        \fi
\providecommand\bibfield[2]{#2}
\providecommand\bibinfo[2]{#2}
\providecommand\natexlab[1]{#1}
\providecommand\showeprint[2][]{arXiv:#2}

\bibitem[Acar et~al\mbox{.}(2019)]%
        {acar2019batchconnect}
\bibfield{author}{\bibinfo{person}{Umut~A. Acar}, \bibinfo{person}{Daniel
  Anderson}, \bibinfo{person}{Guy~E. Blelloch}, {and} \bibinfo{person}{Laxman
  Dhulipala}.} \bibinfo{year}{2019}\natexlab{}.
\newblock \showarticletitle{Parallel Batch-Dynamic Graph Connectivity}. In
  \bibinfo{booktitle}{\emph{{ACM} Symposium on Parallelism in Algorithms and
  Architectures (SPAA)}}. \bibinfo{publisher}{ACM}.
\newblock


\bibitem[Acar et~al\mbox{.}(2020)]%
        {acar2020changeprop}
\bibfield{author}{\bibinfo{person}{Umut~A. Acar}, \bibinfo{person}{Daniel
  Anderson}, \bibinfo{person}{Guy~E. Blelloch}, \bibinfo{person}{Laxman
  Dhulipala}, {and} \bibinfo{person}{Sam Westrick}.}
  \bibinfo{year}{2020}\natexlab{}.
\newblock \showarticletitle{Parallel Batch-Dynamic Trees via Change
  Propagation}. In \bibinfo{booktitle}{\emph{European Symposium on Algorithms
  (ESA)}}.
\newblock


\bibitem[Acar et~al\mbox{.}(2004)]%
        {acar2004dynamizing}
\bibfield{author}{\bibinfo{person}{Umut~A. Acar}, \bibinfo{person}{Guy~E.
  Blelloch}, \bibinfo{person}{Robert Harper}, \bibinfo{person}{Jorge~L.
  Vittes}, {and} \bibinfo{person}{Shan Leung~Maverick Woo}.}
  \bibinfo{year}{2004}\natexlab{}.
\newblock \showarticletitle{Dynamizing static algorithms, with applications to
  dynamic trees and history independence}. In
  \bibinfo{booktitle}{\emph{{ACM-SIAM} Symposium on Discrete Algorithms
  (SODA)}}.
\newblock
\showISBNx{089871558X}


\bibitem[Acar et~al\mbox{.}(2005)]%
        {acar2005experimental}
\bibfield{author}{\bibinfo{person}{Umut~A. Acar}, \bibinfo{person}{Guy~E.
  Blelloch}, {and} \bibinfo{person}{Jorge~L. Vittes}.}
  \bibinfo{year}{2005}\natexlab{}.
\newblock \showarticletitle{An Experimental Analysis of Change Propagation in
  Dynamic Trees}. In \bibinfo{booktitle}{\emph{Proceedings of the Seventh
  Workshop on Algorithm Engineering and Experiments and the Second Workshop on
  Analytic Algorithmics and Combinatorics (ALENEX /ANALCO)}}.
\newblock


\bibitem[Anderson(2023)]%
        {anderson2023thesis}
\bibfield{author}{\bibinfo{person}{Daniel Anderson}.}
  \bibinfo{year}{2023}\natexlab{}.
\newblock \emph{\bibinfo{title}{Parallel Batch-Dynamic Algorithms Dynamic
  Trees, Graphs, and Self-Adjusting Computation}}.
\newblock \bibinfo{thesistype}{Ph.\,D. Dissertation}. \bibinfo{school}{Carnegie
  Mellon University}.
\newblock


\bibitem[Anderson and Blelloch(2023)]%
        {anderson2023deterministic}
\bibfield{author}{\bibinfo{person}{Daniel Anderson} {and}
  \bibinfo{person}{Guy~E. Blelloch}.} \bibinfo{year}{2023}\natexlab{}.
\newblock \bibinfo{title}{Deterministic and Work-Efficient Parallel
  Batch-Dynamic Trees in Low Span}.
\newblock
\showeprint[arxiv]{2306.08786}~[cs.DS]


\bibitem[Anderson and Blelloch(2024)]%
        {anderson2024deterministic}
\bibfield{author}{\bibinfo{person}{Daniel Anderson} {and}
  \bibinfo{person}{Guy~E. Blelloch}.} \bibinfo{year}{2024}\natexlab{}.
\newblock \showarticletitle{Deterministic and Low-Span Work-Efficient Parallel
  Batch-Dynamic Trees}. In \bibinfo{booktitle}{\emph{{ACM} Symposium on
  Parallelism in Algorithms and Architectures (SPAA)}}.
\newblock


\bibitem[Baron(2019)]%
        {baron2019machine}
\bibfield{author}{\bibinfo{person}{Dalya Baron}.}
  \bibinfo{year}{2019}\natexlab{}.
\newblock \bibinfo{title}{Machine Learning in Astronomy: a practical overview}.
\newblock
\showeprint[arxiv]{1904.07248}~[astro-ph.IM]


\bibitem[Bateni et~al\mbox{.}(2024)]%
        {bateni2024s}
\bibfield{author}{\bibinfo{person}{MohammadHossein Bateni},
  \bibinfo{person}{Laxman Dhulipala}, \bibinfo{person}{Kishen~N Gowda},
  \bibinfo{person}{D~Ellis Hershkowitz}, \bibinfo{person}{Rajesh Jayaram},
  {and} \bibinfo{person}{Jakub Lacki}.} \bibinfo{year}{2024}\natexlab{}.
\newblock \showarticletitle{It's Hard to {HAC} Average Linkage!}. In
  \bibinfo{booktitle}{\emph{Intl. Colloq. on Automata, Languages and
  Programming {(ICALP)}}}.
\newblock


\bibitem[Bender and Farach-Colton(2000)]%
        {bender2000lca}
\bibfield{author}{\bibinfo{person}{Michael~A. Bender} {and}
  \bibinfo{person}{Martin Farach-Colton}.} \bibinfo{year}{2000}\natexlab{}.
\newblock \showarticletitle{The LCA problem revisited}. In
  \bibinfo{booktitle}{\emph{Latin American Symposium on Theoretical Informatics
  (LATIN)}}.
\newblock


\bibitem[Bialynicka-Birula and Grossi(2006)]%
        {birula2006amortized}
\bibfield{author}{\bibinfo{person}{Iwona Bialynicka-Birula} {and}
  \bibinfo{person}{Roberto Grossi}.} \bibinfo{year}{2006}\natexlab{}.
\newblock \showarticletitle{Amortized Rigidness in Dynamic Cartesian Trees}. In
  \bibinfo{booktitle}{\emph{Symposium on Theoretical Aspects of Computer
  Science (STACS)}}.
\newblock


\bibitem[Blelloch et~al\mbox{.}(2020)]%
        {blelloch2019optimal}
\bibfield{author}{\bibinfo{person}{Guy~E. Blelloch}, \bibinfo{person}{Jeremy~T.
  Fineman}, \bibinfo{person}{Yan Gu}, {and} \bibinfo{person}{Yihan Sun}.}
  \bibinfo{year}{2020}\natexlab{}.
\newblock \showarticletitle{Optimal Parallel Algorithms in the Binary-Forking
  Model}. In \bibinfo{booktitle}{\emph{{ACM} Symposium on Parallelism in
  Algorithms and Architectures (SPAA)}}.
\newblock


\bibitem[Campello et~al\mbox{.}(2015)]%
        {campello2015hierarchical}
\bibfield{author}{\bibinfo{person}{Ricardo~JGB Campello},
  \bibinfo{person}{Davoud Moulavi}, \bibinfo{person}{Arthur Zimek}, {and}
  \bibinfo{person}{J{\"o}rg Sander}.} \bibinfo{year}{2015}\natexlab{}.
\newblock \showarticletitle{Hierarchical density estimates for data clustering,
  visualization, and outlier detection}.
\newblock \bibinfo{journal}{\emph{{ACM} Transactions on Knowledge Discovery
  from Data ({TKDD})}} \bibinfo{volume}{10}, \bibinfo{number}{1}
  (\bibinfo{year}{2015}).
\newblock


\bibitem[Cole(1988)]%
        {Cole1988}
\bibfield{author}{\bibinfo{person}{Richard Cole}.}
  \bibinfo{year}{1988}\natexlab{}.
\newblock \showarticletitle{Parallel Merge Sort}.
\newblock \bibinfo{journal}{\emph{{SIAM} J. on Computing}}
  \bibinfo{volume}{17}, \bibinfo{number}{4} (\bibinfo{year}{1988}).
\newblock


\bibitem[De~Man et~al\mbox{.}(2024)]%
        {de2024towards}
\bibfield{author}{\bibinfo{person}{Quinten De~Man}, \bibinfo{person}{Laxman
  Dhulipala}, \bibinfo{person}{Adam Karczmarz}, \bibinfo{person}{Jakub
  \L{}\k{a}cki}, \bibinfo{person}{Julian Shun}, {and} \bibinfo{person}{Zhongqi
  Wang}.} \bibinfo{year}{2024}\natexlab{}.
\newblock \showarticletitle{Towards Scalable and Practical Batch-Dynamic
  Connectivity}.
\newblock \bibinfo{journal}{\emph{Proceedings of the VLDB Endowment (PVLDB)}}
  \bibinfo{volume}{18}, \bibinfo{number}{3} (\bibinfo{year}{2024}).
\newblock


\bibitem[Demaine et~al\mbox{.}(2014)]%
        {demaine2014oncartesian}
\bibfield{author}{\bibinfo{person}{Erik~D. Demaine}, \bibinfo{person}{Gad~M.
  Landau}, {and} \bibinfo{person}{Oren Weimann}.}
  \bibinfo{year}{2014}\natexlab{}.
\newblock \showarticletitle{On Cartesian Trees and Range Minimum Queries}.
\newblock \bibinfo{journal}{\emph{Algorithmica}} \bibinfo{volume}{68},
  \bibinfo{number}{3} (\bibinfo{year}{2014}).
\newblock


\bibitem[Dhulipala et~al\mbox{.}(2022a)]%
        {dhulipala2022pac}
\bibfield{author}{\bibinfo{person}{Laxman Dhulipala}, \bibinfo{person}{Guy~E
  Blelloch}, \bibinfo{person}{Yan Gu}, {and} \bibinfo{person}{Yihan Sun}.}
  \bibinfo{year}{2022}\natexlab{a}.
\newblock \showarticletitle{Pac-trees: Supporting parallel and compressed
  purely-functional collections}. In \bibinfo{booktitle}{\emph{Proceedings of
  the 43rd ACM SIGPLAN International Conference on Programming Language Design
  and Implementation}}. \bibinfo{pages}{108--121}.
\newblock


\bibitem[Dhulipala et~al\mbox{.}(2019)]%
        {dhulipala2019low}
\bibfield{author}{\bibinfo{person}{Laxman Dhulipala}, \bibinfo{person}{Guy~E
  Blelloch}, {and} \bibinfo{person}{Julian Shun}.}
  \bibinfo{year}{2019}\natexlab{}.
\newblock \showarticletitle{Low-latency graph streaming using compressed
  purely-functional trees}. In \bibinfo{booktitle}{\emph{ACM Conference on
  Programming Language Design and Implementation (PLDI)}}.
\newblock


\bibitem[Dhulipala et~al\mbox{.}(2024)]%
        {Dhulipala2024Dendrogram}
\bibfield{author}{\bibinfo{person}{Laxman Dhulipala}, \bibinfo{person}{Xiaojun
  Dong}, \bibinfo{person}{Kishen~N. Gowda}, {and} \bibinfo{person}{Yan Gu}.}
  \bibinfo{year}{2024}\natexlab{}.
\newblock \showarticletitle{Optimal Parallel Algorithms for Dendrogram
  Computation and Single-Linkage Clustering}. In
  \bibinfo{booktitle}{\emph{{ACM} Symposium on Parallelism in Algorithms and
  Architectures (SPAA)}}.
\newblock


\bibitem[Dhulipala et~al\mbox{.}(2020)]%
        {dhulipala2019parallel}
\bibfield{author}{\bibinfo{person}{Laxman Dhulipala}, \bibinfo{person}{David
  Durfee}, \bibinfo{person}{Janardhan Kulkarni}, \bibinfo{person}{Richard
  Peng}, \bibinfo{person}{Saurabh Sawlani}, {and} \bibinfo{person}{Xiaorui
  Sun}.} \bibinfo{year}{2020}\natexlab{}.
\newblock \showarticletitle{Parallel Batch-Dynamic Graphs: Algorithms and Lower
  Bounds}. In \bibinfo{booktitle}{\emph{{ACM-SIAM} Symposium on Discrete
  Algorithms (SODA)}}.
\newblock


\bibitem[Dhulipala et~al\mbox{.}(2021a)]%
        {dhulipala2021hierarchical}
\bibfield{author}{\bibinfo{person}{Laxman Dhulipala}, \bibinfo{person}{David
  Eisenstat}, \bibinfo{person}{Jakub {\L}{\k{a}}cki}, \bibinfo{person}{Vahab
  Mirrokni}, {and} \bibinfo{person}{Jessica Shi}.}
  \bibinfo{year}{2021}\natexlab{a}.
\newblock \showarticletitle{Hierarchical agglomerative graph clustering in
  nearly-linear time}. In \bibinfo{booktitle}{\emph{International Conference on
  Machine Learning (ICML)}}. \bibinfo{publisher}{PMLR}.
\newblock


\bibitem[Dhulipala et~al\mbox{.}(2022b)]%
        {parhac}
\bibfield{author}{\bibinfo{person}{Laxman Dhulipala}, \bibinfo{person}{David
  Eisenstat}, \bibinfo{person}{Jakub {\L}{\k{a}}cki}, \bibinfo{person}{Vahab
  Mirrokni}, {and} \bibinfo{person}{Jessica Shi}.}
  \bibinfo{year}{2022}\natexlab{b}.
\newblock \showarticletitle{Hierarchical Agglomerative Graph Clustering in
  Poly-Logarithmic Depth}. In \bibinfo{booktitle}{\emph{Neural Information
  Processing Systems (NeurIPS)}}.
\newblock


\bibitem[Dhulipala et~al\mbox{.}(2023)]%
        {dhulipala2023terahac}
\bibfield{author}{\bibinfo{person}{Laxman Dhulipala}, \bibinfo{person}{Jakub
  {\L}{\k{a}}cki}, \bibinfo{person}{Jason Lee}, {and} \bibinfo{person}{Vahab
  Mirrokni}.} \bibinfo{year}{2023}\natexlab{}.
\newblock \showarticletitle{Terahac: Hierarchical agglomerative clustering of
  trillion-edge graphs}.
\newblock \bibinfo{journal}{\emph{Proceedings of the {ACM} on Management of
  Data}} \bibinfo{volume}{1}, \bibinfo{number}{3} (\bibinfo{year}{2023}).
\newblock


\bibitem[Dhulipala et~al\mbox{.}(2021b)]%
        {dhulipala2021parallel}
\bibfield{author}{\bibinfo{person}{Laxman Dhulipala},
  \bibinfo{person}{Quanquan~C Liu}, \bibinfo{person}{Julian Shun}, {and}
  \bibinfo{person}{Shangdi Yu}.} \bibinfo{year}{2021}\natexlab{b}.
\newblock \showarticletitle{Parallel batch-dynamic k-clique counting}. In
  \bibinfo{booktitle}{\emph{{ACM-SIAM} Symposium on Algorithmic Principles of
  Computer Systems (APOCS)}}.
\newblock


\bibitem[Feigelson and Babu(1998)]%
        {feigelson1998statistical}
\bibfield{author}{\bibinfo{person}{ED Feigelson} {and} \bibinfo{person}{GJ
  Babu}.} \bibinfo{year}{1998}\natexlab{}.
\newblock \showarticletitle{Statistical methodology for large astronomical
  surveys}. In \bibinfo{booktitle}{\emph{Symposium-International Astronomical
  Union}}, Vol.~\bibinfo{volume}{179}. Cambridge University Press.
\newblock


\bibitem[Frederickson(1985)]%
        {frederickson1985data}
\bibfield{author}{\bibinfo{person}{Greg~N Frederickson}.}
  \bibinfo{year}{1985}\natexlab{}.
\newblock \showarticletitle{Data structures for on-line updating of minimum
  spanning trees, with applications}.
\newblock \bibinfo{journal}{\emph{{SIAM} J. on Computing}}
  \bibinfo{volume}{14}, \bibinfo{number}{4} (\bibinfo{year}{1985}).
\newblock


\bibitem[Gasperini et~al\mbox{.}(2019)]%
        {gasperini2019genome}
\bibfield{author}{\bibinfo{person}{Molly Gasperini}, \bibinfo{person}{Andrew~J
  Hill}, \bibinfo{person}{Jos{\'e}~L McFaline-Figueroa}, \bibinfo{person}{Beth
  Martin}, \bibinfo{person}{Seungsoo Kim}, \bibinfo{person}{Melissa~D Zhang},
  \bibinfo{person}{Dana Jackson}, \bibinfo{person}{Anh Leith},
  \bibinfo{person}{Jacob Schreiber}, \bibinfo{person}{William~S Noble},
  {et~al\mbox{.}}} \bibinfo{year}{2019}\natexlab{}.
\newblock \showarticletitle{A genome-wide framework for mapping gene regulation
  via cellular genetic screens}.
\newblock \bibinfo{journal}{\emph{Cell}} \bibinfo{volume}{176},
  \bibinfo{number}{1} (\bibinfo{year}{2019}).
\newblock


\bibitem[G{\"o}tz et~al\mbox{.}(2018)]%
        {gotz2018parallel}
\bibfield{author}{\bibinfo{person}{Markus G{\"o}tz}, \bibinfo{person}{Gabriele
  Cavallaro}, \bibinfo{person}{Thierry G{\'e}raud}, \bibinfo{person}{Matthias
  Book}, {and} \bibinfo{person}{Morris Riedel}.}
  \bibinfo{year}{2018}\natexlab{}.
\newblock \showarticletitle{Parallel computation of component trees on
  distributed memory machines}.
\newblock \bibinfo{journal}{\emph{{IEEE} Transactions on Parallel and
  Distributed Systems}} \bibinfo{volume}{29}, \bibinfo{number}{11}
  (\bibinfo{year}{2018}).
\newblock


\bibitem[Gower and Ross(1969)]%
        {Gower1969MST}
\bibfield{author}{\bibinfo{person}{J.~C. Gower} {and} \bibinfo{person}{G.~J.~S.
  Ross}.} \bibinfo{year}{1969}\natexlab{}.
\newblock \showarticletitle{Minimum Spanning Trees and Single Linkage Cluster
  Analysis}.
\newblock \bibinfo{journal}{\emph{Journal of the Royal Statistical Society.
  Series C (Applied Statistics)}} \bibinfo{volume}{18}, \bibinfo{number}{1}
  (\bibinfo{year}{1969}).
\newblock


\bibitem[Havel et~al\mbox{.}(2019)]%
        {havel2019efficient}
\bibfield{author}{\bibinfo{person}{Ji{\v{r}}{\'\i} Havel},
  \bibinfo{person}{Fran{\c{c}}ois Merciol}, {and}
  \bibinfo{person}{S{\'e}bastien Lef{\`e}vre}.}
  \bibinfo{year}{2019}\natexlab{}.
\newblock \showarticletitle{Efficient tree construction for multiscale image
  representation and processing}.
\newblock \bibinfo{journal}{\emph{Journal of Real-Time Image Processing}}
  \bibinfo{volume}{16} (\bibinfo{year}{2019}).
\newblock


\bibitem[Henry et~al\mbox{.}(2005)]%
        {henry2005cluster}
\bibfield{author}{\bibinfo{person}{David~B Henry}, \bibinfo{person}{Patrick~H
  Tolan}, {and} \bibinfo{person}{Deborah Gorman-Smith}.}
  \bibinfo{year}{2005}\natexlab{}.
\newblock \showarticletitle{Cluster analysis in family psychology research.}
\newblock \bibinfo{journal}{\emph{Journal of Family Psychology}}
  \bibinfo{volume}{19}, \bibinfo{number}{1} (\bibinfo{year}{2005}).
\newblock


\bibitem[Henzinger and King(1995)]%
        {henzinger1995randomized}
\bibfield{author}{\bibinfo{person}{Monika~Rauch Henzinger} {and}
  \bibinfo{person}{Valerie King}.} \bibinfo{year}{1995}\natexlab{}.
\newblock \showarticletitle{Randomized dynamic graph algorithms with
  polylogarithmic time per operation}. In \bibinfo{booktitle}{\emph{{ACM}
  Symposium on Theory of Computing (STOC)}}.
\newblock


\bibitem[Holm et~al\mbox{.}(2001)]%
        {holm2001poly}
\bibfield{author}{\bibinfo{person}{Jacob Holm}, \bibinfo{person}{Kristian
  De~Lichtenberg}, {and} \bibinfo{person}{Mikkel Thorup}.}
  \bibinfo{year}{2001}\natexlab{}.
\newblock \showarticletitle{Poly-logarithmic deterministic fully-dynamic
  algorithms for connectivity, minimum spanning tree, 2-edge, and
  biconnectivity}.
\newblock \bibinfo{journal}{\emph{J. {ACM}}} \bibinfo{volume}{48},
  \bibinfo{number}{4} (\bibinfo{year}{2001}).
\newblock


\bibitem[J\'{a}J\'{a}(1992)]%
        {jaja1992parallel}
\bibfield{author}{\bibinfo{person}{Joseph J\'{a}J\'{a}}.}
  \bibinfo{year}{1992}\natexlab{}.
\newblock \bibinfo{booktitle}{\emph{An Introduction to Parallel Algorithms}}.
\newblock \bibinfo{publisher}{Addison Wesley Longman Publishing Co., Inc.},
  \bibinfo{address}{USA}.
\newblock
\showISBNx{0201548569}


\bibitem[Letunic and Bork(2007)]%
        {letunic2007interactive}
\bibfield{author}{\bibinfo{person}{Ivica Letunic} {and} \bibinfo{person}{Peer
  Bork}.} \bibinfo{year}{2007}\natexlab{}.
\newblock \showarticletitle{Interactive Tree Of Life (iTOL): an online tool for
  phylogenetic tree display and annotation}.
\newblock \bibinfo{journal}{\emph{Bioinformatics}} \bibinfo{volume}{23},
  \bibinfo{number}{1} (\bibinfo{year}{2007}).
\newblock


\bibitem[Liu et~al\mbox{.}(2022)]%
        {liu2022parallel}
\bibfield{author}{\bibinfo{person}{Quanquan~C Liu}, \bibinfo{person}{Jessica
  Shi}, \bibinfo{person}{Shangdi Yu}, \bibinfo{person}{Laxman Dhulipala}, {and}
  \bibinfo{person}{Julian Shun}.} \bibinfo{year}{2022}\natexlab{}.
\newblock \showarticletitle{Parallel batch-dynamic algorithms for k-core
  decomposition and related graph problems}. In
  \bibinfo{booktitle}{\emph{Proceedings of the 34th ACM Symposium on
  Parallelism in Algorithms and Architectures}}. \bibinfo{pages}{191--204}.
\newblock


\bibitem[Manning et~al\mbox{.}(2008)]%
        {irbook}
\bibfield{author}{\bibinfo{person}{Christopher~D Manning},
  \bibinfo{person}{Prabhakar Raghavan}, {and} \bibinfo{person}{Hinrich
  Sch{\"u}tze}.} \bibinfo{year}{2008}\natexlab{}.
\newblock \bibinfo{booktitle}{\emph{Introduction to Information Retrieval}}.
\newblock \bibinfo{publisher}{Cambridge University Press}.
\newblock


\bibitem[Men et~al\mbox{.}(2025)]%
        {men2025parallel}
\bibfield{author}{\bibinfo{person}{Ziyang Men}, \bibinfo{person}{Zheqi Shen},
  \bibinfo{person}{Yan Gu}, {and} \bibinfo{person}{Yihan Sun}.}
  \bibinfo{year}{2025}\natexlab{}.
\newblock \showarticletitle{Parallel kd-tree with Batch Updates}.
\newblock \bibinfo{journal}{\emph{Proceedings of the ACM on Management of
  Data}} \bibinfo{volume}{3}, \bibinfo{number}{1} (\bibinfo{year}{2025}),
  \bibinfo{pages}{1--26}.
\newblock


\bibitem[Miller and Reif(1985)]%
        {miller1985parallel}
\bibfield{author}{\bibinfo{person}{Gary~L Miller} {and} \bibinfo{person}{John~H
  Reif}.} \bibinfo{year}{1985}\natexlab{}.
\newblock \showarticletitle{Parallel tree contraction and its application}. In
  \bibinfo{booktitle}{\emph{{IEEE} Symposium on Foundations of Computer Science
  (FOCS)}}, Vol.~\bibinfo{volume}{26}.
\newblock


\bibitem[Monath et~al\mbox{.}(2023)]%
        {monath2023online}
\bibfield{author}{\bibinfo{person}{Nicholas Monath}, \bibinfo{person}{Manzil
  Zaheer}, {and} \bibinfo{person}{Andrew McCallum}.}
  \bibinfo{year}{2023}\natexlab{}.
\newblock \showarticletitle{Online level-wise hierarchical clustering}. In
  \bibinfo{booktitle}{\emph{{SIGKDD} Conference on Knowledge Discovery and Data
  Mining (KDD)}}.
\newblock


\bibitem[Nolet et~al\mbox{.}(2023)]%
        {nolet2023cuslink}
\bibfield{author}{\bibinfo{person}{Corey~J Nolet}, \bibinfo{person}{Divye
  Gala}, \bibinfo{person}{Alex Fender}, \bibinfo{person}{Mahesh Doijade},
  \bibinfo{person}{Joe Eaton}, \bibinfo{person}{Edward Raff},
  \bibinfo{person}{John Zedlewski}, \bibinfo{person}{Brad Rees}, {and}
  \bibinfo{person}{Tim Oates}.} \bibinfo{year}{2023}\natexlab{}.
\newblock \showarticletitle{cuSLINK: Single-linkage Agglomerative Clustering on
  the GPU}. In \bibinfo{booktitle}{\emph{Joint European Conference on Machine
  Learning and Knowledge Discovery in Databases}}.
\newblock


\bibitem[Ouzounis and Soille(2012)]%
        {ouzounis2012alpha}
\bibfield{author}{\bibinfo{person}{Georgios~K Ouzounis} {and}
  \bibinfo{person}{Pierre Soille}.} \bibinfo{year}{2012}\natexlab{}.
\newblock \showarticletitle{The alpha-tree algorithm}.
\newblock \bibinfo{journal}{\emph{JRC Scientific and Policy Report}}
  (\bibinfo{year}{2012}).
\newblock


\bibitem[Patra\c{s}cu and Demaine(2004)]%
        {patrascu2004lower}
\bibfield{author}{\bibinfo{person}{Mihai Patra\c{s}cu} {and}
  \bibinfo{person}{Erik~D. Demaine}.} \bibinfo{year}{2004}\natexlab{}.
\newblock \showarticletitle{Lower bounds for dynamic connectivity}. In
  \bibinfo{booktitle}{\emph{{ACM} Symposium on Theory of Computing (STOC)}}.
\newblock


\bibitem[Sao et~al\mbox{.}(2024)]%
        {sao2024pandora}
\bibfield{author}{\bibinfo{person}{Piyush Sao}, \bibinfo{person}{Andrey
  Prokopenko}, {and} \bibinfo{person}{Damien Lebrun-Grandie}.}
  \bibinfo{year}{2024}\natexlab{}.
\newblock \showarticletitle{PANDORA: A Parallel Dendrogram Construction
  Algorithm for Single Linkage Clustering on GPU}. In
  \bibinfo{booktitle}{\emph{International Conference on Parallel Processing
  (ICPP)}}.
\newblock


\bibitem[Shun and Blelloch(2014)]%
        {SB14cartesian}
\bibfield{author}{\bibinfo{person}{Julian Shun} {and} \bibinfo{person}{Guy~E.
  Blelloch}.} \bibinfo{year}{2014}\natexlab{}.
\newblock \showarticletitle{A Simple Parallel Cartesian Tree Algorithm and Its
  Application to Parallel Suffix Tree Construction}.
\newblock \bibinfo{journal}{\emph{{ACM} Transactions on Parallel Computing
  (TOPC)}} \bibinfo{volume}{1}, \bibinfo{number}{1} (\bibinfo{year}{2014}).
\newblock


\bibitem[Sleator and Tarjan(1983)]%
        {sleator1983data}
\bibfield{author}{\bibinfo{person}{Daniel~D Sleator} {and}
  \bibinfo{person}{Robert~Endre Tarjan}.} \bibinfo{year}{1983}\natexlab{}.
\newblock \showarticletitle{A data structure for dynamic trees}.
\newblock \bibinfo{journal}{\emph{J. Computer and System Sciences}}
  \bibinfo{volume}{26}, \bibinfo{number}{3} (\bibinfo{year}{1983}).
\newblock


\bibitem[Tseng et~al\mbox{.}(2019)]%
        {tseng2019batch}
\bibfield{author}{\bibinfo{person}{Thomas Tseng}, \bibinfo{person}{Laxman
  Dhulipala}, {and} \bibinfo{person}{Guy Blelloch}.}
  \bibinfo{year}{2019}\natexlab{}.
\newblock \showarticletitle{Batch-parallel euler tour trees}. In
  \bibinfo{booktitle}{\emph{2019 Proceedings of the Twenty-First Workshop on
  Algorithm Engineering and Experiments (ALENEX)}}. SIAM,
  \bibinfo{pages}{92--106}.
\newblock


\bibitem[Tseng et~al\mbox{.}(2022)]%
        {tseng2022parallel}
\bibfield{author}{\bibinfo{person}{Tom Tseng}, \bibinfo{person}{Laxman
  Dhulipala}, {and} \bibinfo{person}{Julian Shun}.}
  \bibinfo{year}{2022}\natexlab{}.
\newblock \showarticletitle{Parallel Batch-Dynamic Minimum Spanning Forest and
  the Efficiency of Dynamic Agglomerative Graph Clustering}. In
  \bibinfo{booktitle}{\emph{{ACM} Symposium on Parallelism in Algorithms and
  Architectures (SPAA)}}.
\newblock


\bibitem[Vuillemin(1980)]%
        {vuillemin1980unifying}
\bibfield{author}{\bibinfo{person}{Jean Vuillemin}.}
  \bibinfo{year}{1980}\natexlab{}.
\newblock \showarticletitle{A unifying look at data structures}.
\newblock \bibinfo{journal}{\emph{Commun. {ACM}}} \bibinfo{volume}{23},
  \bibinfo{number}{4} (\bibinfo{year}{1980}).
\newblock


\bibitem[Wang et~al\mbox{.}(2020)]%
        {wang2020theoretically}
\bibfield{author}{\bibinfo{person}{Yiqiu Wang}, \bibinfo{person}{Yan Gu}, {and}
  \bibinfo{person}{Julian Shun}.} \bibinfo{year}{2020}\natexlab{}.
\newblock \showarticletitle{Theoretically-Efficient and Practical Parallel
  {DBSCAN}}. In \bibinfo{booktitle}{\emph{ACM SIGMOD International Conference
  on Management of Data (SIGMOD)}}.
\newblock


\bibitem[Wang et~al\mbox{.}(2021a)]%
        {wang2021fast}
\bibfield{author}{\bibinfo{person}{Yiqiu Wang}, \bibinfo{person}{Shangdi Yu},
  \bibinfo{person}{Yan Gu}, {and} \bibinfo{person}{Julian Shun}.}
  \bibinfo{year}{2021}\natexlab{a}.
\newblock \showarticletitle{Fast Parallel Algorithms for Euclidean Minimum
  Spanning Tree and Hierarchical Spatial Clustering}. In
  \bibinfo{booktitle}{\emph{ACM SIGMOD International Conference on Management
  of Data (SIGMOD)}}.
\newblock


\bibitem[Wang et~al\mbox{.}(2021b)]%
        {wang2020closest}
\bibfield{author}{\bibinfo{person}{Yiqiu Wang}, \bibinfo{person}{Shangdi Yu},
  \bibinfo{person}{Yan Gu}, {and} \bibinfo{person}{Julian Shun}.}
  \bibinfo{year}{2021}\natexlab{b}.
\newblock \bibinfo{title}{A Parallel Batch-Dynamic Data Structure for the
  Closest Pair Problem}.
\newblock
\showeprint[arxiv]{2010.02379}~[cs.DS]


\bibitem[Wulff-Nilsen(2017)]%
        {wulff2017fully}
\bibfield{author}{\bibinfo{person}{Christian Wulff-Nilsen}.}
  \bibinfo{year}{2017}\natexlab{}.
\newblock \showarticletitle{Fully-dynamic minimum spanning forest with improved
  worst-case update time}. In \bibinfo{booktitle}{\emph{{ACM} Symposium on
  Theory of Computing (STOC)}}.
\newblock


\bibitem[Yao(1981)]%
        {cellprobe}
\bibfield{author}{\bibinfo{person}{Andrew Chi-Chih Yao}.}
  \bibinfo{year}{1981}\natexlab{}.
\newblock \showarticletitle{Should Tables Be Sorted?}
\newblock \bibinfo{journal}{\emph{J. {ACM}}} \bibinfo{volume}{28},
  \bibinfo{number}{3} (\bibinfo{year}{1981}).
\newblock


\bibitem[Yengo et~al\mbox{.}(2022)]%
        {yengo2022saturated}
\bibfield{author}{\bibinfo{person}{Lo{\"\i}c Yengo}, \bibinfo{person}{Sailaja
  Vedantam}, \bibinfo{person}{Eirini Marouli}, \bibinfo{person}{Julia
  Sidorenko}, \bibinfo{person}{Eric Bartell}, \bibinfo{person}{Saori Sakaue},
  \bibinfo{person}{Marielisa Graff}, \bibinfo{person}{Anders~U Eliasen},
  \bibinfo{person}{Yunxuan Jiang}, \bibinfo{person}{Sridharan Raghavan},
  {et~al\mbox{.}}} \bibinfo{year}{2022}\natexlab{}.
\newblock \showarticletitle{A saturated map of common genetic variants
  associated with human height}.
\newblock \bibinfo{journal}{\emph{Nature}} \bibinfo{volume}{610},
  \bibinfo{number}{7933} (\bibinfo{year}{2022}).
\newblock


\bibitem[Yesantharao et~al\mbox{.}(2021)]%
        {yesantharao2021parallel}
\bibfield{author}{\bibinfo{person}{Rahul Yesantharao}, \bibinfo{person}{Yiqiu
  Wang}, \bibinfo{person}{Laxman Dhulipala}, {and} \bibinfo{person}{Julian
  Shun}.} \bibinfo{year}{2021}\natexlab{}.
\newblock \bibinfo{title}{Parallel Batch-Dynamic $k$d-Trees}.
\newblock
\showeprint[arxiv]{2112.06188}~[cs.DS]


\bibitem[Yim and Ramdeen(2015)]%
        {yim2015hierarchical}
\bibfield{author}{\bibinfo{person}{Odilia Yim} {and} \bibinfo{person}{Kylee~T
  Ramdeen}.} \bibinfo{year}{2015}\natexlab{}.
\newblock \showarticletitle{Hierarchical cluster analysis: comparison of three
  linkage measures and application to psychological data}.
\newblock \bibinfo{journal}{\emph{The quantitative methods for psychology}}
  \bibinfo{volume}{11}, \bibinfo{number}{1} (\bibinfo{year}{2015}).
\newblock


\bibitem[Yu et~al\mbox{.}(2025)]%
        {yu2025dynhac}
\bibfield{author}{\bibinfo{person}{Shangdi Yu}, \bibinfo{person}{Laxman
  Dhulipala}, \bibinfo{person}{Jakub Łącki}, {and} \bibinfo{person}{Nikos
  Parotsidis}.} \bibinfo{year}{2025}\natexlab{}.
\newblock \bibinfo{title}{DynHAC: Fully Dynamic Approximate Hierarchical
  Agglomerative Clustering}.
\newblock
\showeprint[arxiv]{2501.07745}~[cs.DS]


\bibitem[Yu et~al\mbox{.}(2021)]%
        {yu2021parchain}
\bibfield{author}{\bibinfo{person}{Shangdi Yu}, \bibinfo{person}{Yiqiu Wang},
  \bibinfo{person}{Yan Gu}, \bibinfo{person}{Laxman Dhulipala}, {and}
  \bibinfo{person}{Julian Shun}.} \bibinfo{year}{2021}\natexlab{}.
\newblock \showarticletitle{ParChain: A Framework for Parallel Hierarchical
  Agglomerative Clustering using Nearest-Neighbor Chain}.
\newblock \bibinfo{journal}{\emph{Proceedings of the VLDB Endowment (PVLDB)}}
  (\bibinfo{year}{2021}).
\newblock


\end{thebibliography}
\clearpage


\end{document}